\documentclass{article}

\usepackage{arxiv}

\usepackage[utf8]{inputenc} % allow utf-8 input
\usepackage[T1]{fontenc}    % use 8-bit T1 fonts
\usepackage{url}            % simple URL typesetting
\usepackage{booktabs}       % professional-quality tables
\usepackage{amsmath}
\usepackage{amsthm}
\usepackage{amssymb}
\usepackage{hyperref}
\hypersetup{
  colorlinks   = true,    % Colours links instead of ugly boxes
  urlcolor     = blue,    % Colour for external hyperlinks
  linkcolor    = blue,    % Colour of internal links
  citecolor    = red      % Colour of citations
}
\usepackage{xcolor, soul}
\usepackage{tikz}
\usepackage{multicol}
\usepackage[ruled,vlined,noline,linesnumbered]{algorithm2e}
\usepackage{cite}
\usepackage{floatrow}
\usepackage{setspace}
\usepackage{multirow}
\usepackage{geometry}
\usepackage{amsfonts}       % blackboard math symbols
\usepackage{nicefrac}       % compact symbols for 1/2, etc.
\usepackage{microtype}      % microtypography
\usepackage{lipsum}		% Can be removed after putting your text content
\usepackage{graphicx}
\usepackage{doi}
\usepackage[most]{tcolorbox} %colored boxes in equation mode
\tcbset{colback=white!10!white, colframe=blue!50!black, 
        highlight math style= {enhanced, %<-- needed for the ?remember? options
            colframe=blue,colback=blue!10!white,boxsep=0pt}
        }
\allowdisplaybreaks 
\definecolor{myyellow}{RGB}{201,177,12}
\definecolor{myorange}{RGB}{211,126,33}
\definecolor{myred}{RGB}{226,74,74}

\newtheorem{theorem}{Theorem}
\newtheorem{definition}{Definition}
\newtheorem{corollary}{Corollary}
\newtheorem{lemma}{Lemma}

\title{An Improvement of Reed's Treewidth Approximation}

\author{ \href{https://sites.google.com/view/belbasi/home} {Mahdi Belbasi}\\
	Department of Computer Science and Engineering\\
	The Pennsylvania State University\\
	University Park, PA\\
	\href{mailto:belbasi@psu.edu}{\texttt{belbasi@psu.edu}} \\
	\And
	\href{https://www.cse.psu.edu/~fhs/} {Martin F\"urer}\\
	Department of Computer Science and Engineering\\
	The Pennsylvania State University\\
	University Park, PA\\
	\href{mailto:fhs@psu.edu}{\texttt{fhs@psu.edu}} \\
}

\renewcommand{\headerleft}{M. Belbasi, M. F\"urer}
\renewcommand{\shorttitle}{An Improvement of Reed's Treewidth Approximation}

\hypersetup{
pdftitle={A template for the arxiv style},
pdfsubject={q-bio.NC, q-bio.QM},
pdfauthor={David S.~Hippocampus, Elias D.~Striatum},
pdfkeywords={First keyword, Second keyword, More},
}

\begin{document}
\maketitle

\begin{abstract}
We present a new approximation algorithm for the treewidth problem which finds an upper bound on the treewidth and constructs a corresponding tree decomposition as well. Our algorithm is a faster variation of Reed's classical algorithm. For the benefit of the reader, and to be able to compare these two algorithms, we start with a detailed time analysis of Reed's algorithm. We fill in many details that have been omitted in Reed's paper. Computing tree decompositions parameterized by the treewidth $k$ is fixed parameter tractable (FPT), meaning that there are algorithms running in time $\mathcal{O}(f(k) g(n))$ where $f$ is a computable function, and $g(n)$ is polynomial in $n$, where $n$ is the number of vertices. An analysis of Reed's algorithm shows $f(k) = 2^{\mathcal{O}(k \log k)}$ and $g(n) = n \log n$ for a 5-approximation. Reed simply claims time $\mathcal{O}(n \log n)$ for bounded $k$ for his constant factor approximation algorithm, but the bound of $2^{\Omega(k \log k)} n \log n$ is well known. From a practical point of view, we notice that the time of Reed's algorithm also contains a term of $\mathcal{O}(k^2 2^{24k} n \log n)$, which for small $k$ is much worse than the asymptotically leading term of $2^{\mathcal{O}(k \log k)} n \log n$. We analyze $f(k)$ more precisely, because the purpose of this paper is to improve the running times for all reasonably small values of $k$.

Our algorithm runs in $\mathcal{O}(f(k)n\log{n})$ too, but with a much smaller dependence on $k$. In our case, $f(k) = 2^{\mathcal{O}(k)}$. This algorithm is simple and fast, especially for small values of $k$. We should mention that Bodlaender et al.\ [2016] have an algorithm with a linear dependence on $n$, and Korhonen [2021] obtains the much better approximation ratio of 2, while the current paper achieves a better dependence on $k$.

\end{abstract}

% The content of the paper goes here

\section{Introduction}
Since the 1970s and early 1980s, when the notions of treewidth and tree decomposition were introduced \cite{bertele1973non, halin1976s, robertson1984graph}, they have played important roles in computer science \cite{bodlaender2005discovering}. In a nutshell, treewidth is a parameter of a graph that measures how similar it is to a tree. One of the main reasons that the tree decomposition is widely studied is that many NP-complete problems have efficient algorithms for graphs with small treewidth. In fact, Courcelle's metatheorem \cite{courcelle1990monadic} states that every graph property definable in monadic second-order logic of graphs can be solved in linear time on graphs of bounded treewidth. The first step of solving such problems is to find an optimal or near-optimal tree decomposition. However, finding an optimal tree decomposition itself is NP-hard \cite{ArnborgCP87}. In this work, we propose an algorithm which is based on Reed's algorithm \cite{reed1992finding} to approximate the treewidth and find an approximately optimal tree decomposition. 

\begin{definition}
A graph problem is \emph{fixed parameter tractable (FPT)} if it can be solved in time $\mathcal{O}\left(f\left(k\right)n^{\mathcal{O}\left(1\right)}\right)$, where $f$ is a computable function, $k$ is a parameter of the graph, and $n$ is the input graph size. 
\end{definition}
\subsection{Previously Known Results}
In this work, we are interested in algorithms which run fast (polynomial in terms of the number of vertices) for graphs with bounded treewidth. One of the first algorithms given for this problem goes back to the same paper where treewidth has been shown to be NP-complete. Arnborg et al.\  \cite{ArnborgCP87} gave an algorithm which runs in time $\mathcal{O}(n^{k+2})$. 
%%Later on, research started on searching for FPT algorithms for this problem. In our case, the parameter is treewidth. Robertson and Seymour showed that the treewidth is an FPT graph problem \cite{robertson1986graph2, }.  
In 1995, Robertson and Seymour \cite{robertson1995graph} gave a quadratic time FPT approximation algorithm. Later on, Lagergren introduced an 8-approximation algorithm with the time complexity of $2^{\mathcal{O}(k\log{k})}n\log^2{n}$ \cite{lagergren1996efficient}.
%% which runs in $\mathcal{O}(f(k)n^2)$. In 90s, this problem got lots of attentions. 
In 1992, Reed \cite{reed1992finding} improved these algorithms to have an algorithm
%%(although he did not mention the approximation ratio) and he claimed that it runs in $\mathcal{O}(n\log(n))$, for fixed $k$. 
running in time $2^{\mathcal{O}\left(k \log k\right)}n\log n$. 
In this paper, we formally show that the approximation ratio of Reed's algorithm is 7 or 5, depending on the frequency of the split by volume. We show that this algorithm runs in time $\mathcal{O}\left(2^{24k}k! n\log{n}\right)$, in order to be able to compare it to our algorithm. 
%%Reed's algorithm is based on \cite{robertson1995graph}. Both of these algorithms are
Like the algorithm of Robertson and Seymour, Reed's algorithm is recursive. In \cite{robertson1995graph}, they find a separator that partitions $G$ into two parts but they do not force the separator to partition the entire graph in a balanced fashion.
%That is why they get $\mathcal{O}(n^2)$ for time complexity. 
Reed finds a separator which partitions the graph in a balanced way to obtain time $\mathcal{O}(n\log{n})$ for bounded $k$. This paper 
focuses on this algorithm.
%%is the one that we focus here. To get the big picture, we continue with the other available algorithms.  In 1996, 
Later, Bodlaender gave an exact algorithm which runs in $2^{\mathcal{O}(k^3)}n$ \cite{bodlaender1996linear}.
%%As you see, this algorithm has 
 Even though we focus only on constant-factor approximation algorithms, it is worth mentioning the $(\log k)$-approximation algorithm by Amir~\cite{amir2010approximation} and the $\sqrt{\log k}$-approximation algorithm  by Feige et al.~\cite{FeigeHL2008}.  Later in 2016, Bodlaender et al.\ \cite{bodlaender2016c} gave two constant factor approximation algorithms which run in $2^{\mathcal{O}(k)}\mathcal{O}(n\log{n})$ and $2^{\mathcal{O}(k)}\mathcal{O}(n)$ respectively. The former is a 3-approximation and the latter is af 5-approximation. Although it is a great result from a theoretical point of view, it uses a sophisticated data structure and the constant factor hidden in $\mathcal{O}(k)$ in the exponent is not claimed to be practical. Very recently, Korhonen gave a 2-approximation algorithm for the same problem running in time $2^{\mathcal{O}(k)}n$\cite{korhonen2021single}. He first provides a loose upper bound on the hidden coefficient of $k$ in the exponent. Then he improves it to $10.7549$ by decreasing his potential function to what seems to be a natural barrier in the worst-case scenario. Here, we sacrifice the linear dependence on $n$, but drop the coefficient of $k$ in the exponent to only $7.61$. Table~\ref{tab:history} summarizes the history of previously known algorithms for treewidth problem.
 
\begin{table}[]\label{tab:history}
\resizebox{\columnwidth}{!}{\begin{tabular}{|l|c|cc|c|}
\hline
\multicolumn{1}{|c|}{Reference} & \begin{tabular}[c]{@{}c@{}}Approximation\\ Ratio\end{tabular} & \multicolumn{1}{c|}{\begin{tabular}[c]{@{}c@{}}Dependence\\ on $k$\end{tabular}} & \begin{tabular}[c]{@{}c@{}}Dependence\\ on $n$\end{tabular} & Comments                                                                                                                   \\ \hline
Arnborg et al.\ (1987)~\cite{arnborg1988problems} & 1   &\multicolumn{2}{c|}{$\mathcal{O}\left(n^{k+2}\right)$}  & not FPT \\ \hline
Robertson \& Seymour (1995)~\cite{robertson1995graph,reed_1997}
& 4                                                             & \multicolumn{1}{c|}{$\mathcal{O}\left(3^{3k}\right)$}  & $n^2$ & \\ \hline
Lagergren (1996)~\cite{lagergren1996efficient} & 8 & \multicolumn{1}{c|}{$2^{\mathcal{O}\left(k \log k\right)}$} & 
$n \log^2n$&\\ \hline
Reed (1992)~\cite{reed1992finding}& 5 (or 7)& \multicolumn{1}{c|}{$2^{24k}k!$ (Sec.~\ref{sec:analysisReed})}&$\mathcal{O}\left(n \log n\right)$&    \\ \hline
Bodlaender (1996)~\cite{bodlaender1996linear} & 1 & \multicolumn{1}{c|}{$2^{\mathcal{O}\left(k^3\right)}$}   & $n$ &\\ \hline
Feige et al.\ (2008)    &  $\mathcal{O}(\sqrt{\log k})$    & \multicolumn{1}{c|}{$\mathcal{O}(1)$}   &  $n^{\mathcal{O}(1)}$      & \begin{tabular}[c]{@{}c@{}}not a constant-\\ factor approxi-\\ mation\end{tabular}                                         \\ \hline

\multirow{4}{*}{Amir (2010)~\cite{amir2010approximation}} &$\mathcal{O}( \log k)$& \multicolumn{1}{c|}{$\mathcal{O}(k \log k)$} &  $n^4$ & \begin{tabular}[c]{@{}c@{}}not a constant-\\ factor approxi-\\ mation\end{tabular}                                         \\  \cline{2-5} 
     &     $4$   & \multicolumn{1}{c|}{$\mathcal{O}\left(2^{4.38k}k\right)$ }    &    $n^2$  & \multirow{2}{*}{}          \\ \cline{2-4}& $4.5$ & \multicolumn{1}{c|}{$\mathcal{O}\left(2^{3k}k^{1.5}\right)$ }  & $n^2$ &\\
     \cline{2-4}& $\frac{11}{3}$ & \multicolumn{1}{c|}{$\mathcal{O}\left(2^{3.6982k}k^{3}\right)$ }  & $n^3\log^4 n$ &
     \\ \hline
Fomin et al. (2015)~\cite{fomin2015large}& 1 & \multicolumn{1}{c|}{$\mathcal{O}(1)$}& $1.7347^n$ & not FPT\\ \hline
\multirow{2}{*}{Bodlaender et al. (2016)~\cite{bodlaender2016c}}     &  $3$    & \multicolumn{1}{c|}{$2^{\mathcal{O}(k)}$} & $n\log n$  &   the coefficients of  \\ \cline{2-4}    & $5$ & \multicolumn{1}{c|}{$2^{\mathcal{O}(k)}$}& $n$ & $k$ are not stated\\ \hline
Belbasi \& F\"urer (2021)~\cite{belbasi2021improvement}& $5$ & \multicolumn{1}{c|}{$2^{8.766k}$}  & $n \log n$&\\ \hline
Korhonen (2021)~\cite{korhonen2021single}&  $2$ & \multicolumn{1}{c|}{$2^{10.7549k}$}    & $n$ & \begin{tabular}[c]{@{}c@{}}relatively low\\coefficient of $k$\\ in the exponent,\\good approxi-\\ mation ratio, and\\ only linear\end{tabular} \\ \hline
This paper&5 & \multicolumn{1}{c|}{$2^{7.61k}$} & $n \log n$ &  \begin{tabular}[c]{@{}c@{}}extra $\log n$ \\in the running\\time compared \\  to~\cite{korhonen2021single}\\but has smaller\\ coefficient of $k$\\ in the exponent\end{tabular}\\ \hline
\end{tabular}}\caption{The history of previous algorithms for the treewidth approximation problem}
\end{table}

\subsection{Our Contribution}
First in Section~\ref{sec:analysisReed}, we analyze Reed's algorithm \cite{reed1992finding} in detail. Reed has focused on the dependence on $n$ because his goal was to come up with an  $\mathcal{O}(n\log n)$-time algorithm, for fixed $k$. The ``fastest''\footnote{Considering the dependence on both $n$ and $k$.} algorithm at that time, was Lagergren's $\mathcal{O}(n \log^2 n)$-time algorithm \cite{lagergren1996efficient}. We show that the dependence on $k$ in Reed's algorithm is $2^{\mathcal{O}(k\log{k})}$. Furthermore, we give a proof for the approximation ratio of Reed's algorithm by filling in the details.

Then, we propose two improvements and prove that the approximation ratio stays at $5$. One of our improvements focuses on the notion of a ``balanced split''. We call a split balanced, if we get two parts of volume $1-\epsilon$ and $\epsilon$ (or better). Then, the running time of our algorithm has a factor of $1/\epsilon$. For instance, if we set $\epsilon = \frac{1}{100}$, a generous estimation shows that the dependence on $k$ in our $\mathcal{O}(f(k)n\log{n})$-time algorithm is $k^2\, 2^{8.87k}$, instead of $2^{24k}(k+1)!$ in Reed's algorithm.  Here the asymptotic notation is a bit misleading from a practical point of view, as $2^{24k} = o(k!)$, even though $k!$ is reasonable for small $k$, while $2^{24k}$ is not. Then, we further improve the running time upper bound and show that our algorithm runs in time $\mathcal{O}(2^{7.61k}k^2 n\log n)$. In the end, the main aim of this paper is to produce a constant-factor approximation algorithm that runs in time $2^{ck}n\log{n}$ with $c$ as small as possible.

%%Sectiion 2 removed (after \end{document}   Inserted again. 

\section{Preliminaries}
\subsection{Tree Decomposition}
\begin{definition}
	A \emph{tree decomposition} of a graph $G = (V, E)$, is a tree $\mathcal{T} = (V_{\mathcal{T}}, E_{\mathcal{T}})$ such that each node $x$ in $V_{\mathcal{T}}$ is associated with a set $B_x$ (called the bag of $x$) of vertices in $G$, and such that $\mathcal{T}$ has the following properties:
	\begin{itemize}
		\item The union of all bags is equal to $V.$ In other words, for each $v\in V,$ there exists at least one node $x\in V_{\mathcal{T}}$ with $B_x$ containing $v$.
		
		\item For every edge $\{u, v\} \in E$, there exists a node $x$ such that $u,v\in B_x.$ 
		
		\item For any nodes $x, y \in V_{\mathcal{T}}$, and any node $z\in V_{\mathcal{T}}$ belonging to the path connecting $x$ and $y$ in $\mathcal{T}$, $B_x \cap B_y \subseteq B_z.$
	\end{itemize}
\end{definition}

 In this paper we use a variation of tree decomposition where the adjacent bags differ in at most one vertex (converting can happen in linear time). 

The \emph{width of a tree decomposition} is the size of a largest bag minus one. The \textit{treewidth} of a graph $G$ is the minimum width over all tree decompositions of $G$ called $tw(G)$. Observe that the treewidth of a tree is 1.
 In the following, we reserve the letter $k$ for the treewidth$+1$.

We have to mention that Bodlaender et al.\ \cite{bodlaender2016c} filled in some details on Reed's algorithm. We need to be more detailed because we do not use Reed's algorithm as a black box. That is why we first analyze Reed's algorithm precisely (Section~\ref{sec:analysisReed}), fill in the blanks, and then introduce our improvements of his algorithm (Section \ref{sec:ours}).

\section{Analysis of Reed's Algorithm}\label{sec:analysisReed}
In 1992, Reed gave an elegant algorithm \cite{reed1992finding} to either construct a tree decomposition of width less than $7k$ or $5k$ of a given graph $G$, or declare that the treewidth is at least $k$ and output a subgraph which is a bottleneck (no separator of size $\leq k$). 

\begin{figure}
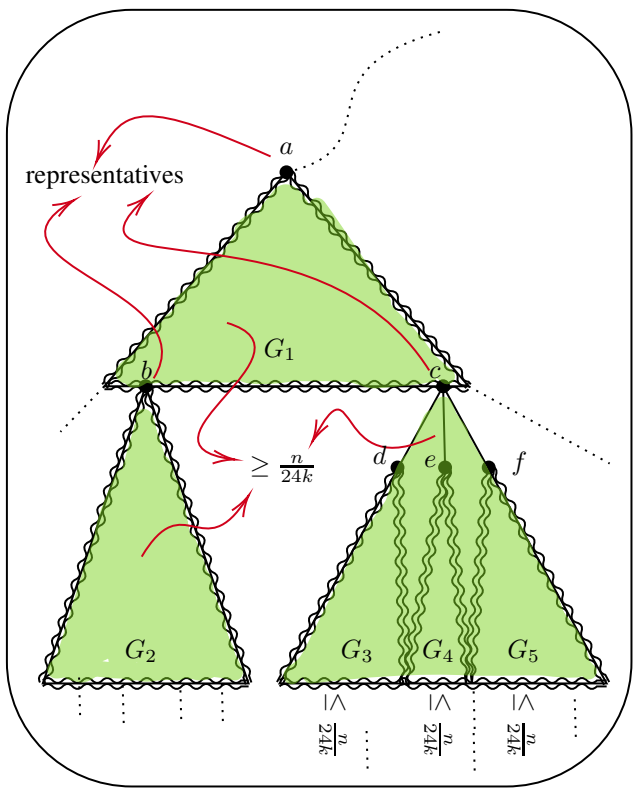

	
	\begin{minipage}[c]{0.70\textwidth}
	\tikzset{every picture/.style={line width=0.75pt}} %set default line width to 0.75pt        
	

	\end{minipage}\hfill
	\begin{minipage}[c]{0.40\textwidth}
	\caption{
		For a given graph $G$, let $V(G)$ be the set of its vertices. Let $G_6 = G[V(G_3) \cup V(G_4) \cup V(G_5) \cup \{c\}]$ (the entire rightmost green subtree rooted at $c$), and assume $|V(G_1)|, |V(G_2)|, \text{ and } |V(G_6)| \geq \frac{n}{24 k}$, and $|V(G_3)|, |V(G_4)|, \text{ and } |V(G_5)| < \frac{n}{24k}$. Here, $d, e,$ and $f$ are NOT representatives but $a, b,$ and $c$ are.
		}\label{fig:reps}
	\end{minipage}
\end{figure}

\subsection{Summary of Reed's Algorithm}
In Reed's algorithm, one of the main tasks is to find a ``balanced'' separator $S$ that splits the graph $G - S$ into two subgraphs with  sets of vertices $X, Y \subseteq V(G)$. Once a balanced separator is found, the algorithm recursively finds a tree decomposition for $G[X\cup S]$ (the subgraph induced by $X \cup S$) and $G[Y \cup S]$.

The main task is to find a balanced separator. Instead of branching on every vertex (going to $X, Y,$ or $S$, which would be exponential in $n$), Reed forms groups of vertices and works with the representatives of the groups. Then, he branches on the representatives. 

Reed does a DFS and finds the deepest vertex $v$ whose subtree has at least $\frac{n}{24k}$ vertices\footnote{Later, we talk about this threshold.} (See Figure \ref{fig:reps}). We call such a vertex a ``representative''. He defines the weight of the representative $v$ as the size of its subtree rooted at $v$ in the DFS, denoted by $w(v)$. The idea here, is that if some representatives with a large total weight go to either $X$ or $Y$, then most of their descendants will go to the same set. The reason is that if a descendant goes to the other side, the path connecting the representative to the descendant has at least one vertex in the separator. However, we know that the separator cannot have more than $k$ vertices. So, not many vertices will go to the wrong set; not more than $\frac{n}{24}$ vertices, in total. This is because every subtree that partially goes to the other side should go through the separator and have one vertex there. Hence, not more than $k$ subtrees rooted at children of the representatives can go through the separator, which results in at most $\frac{n}{24}$ vertices on the wrong side. This nice property allows Reed to work with the set of representatives (which is much smaller) rather than all the vertices (see Figure~\ref{fig:error}).

\begin{figure}
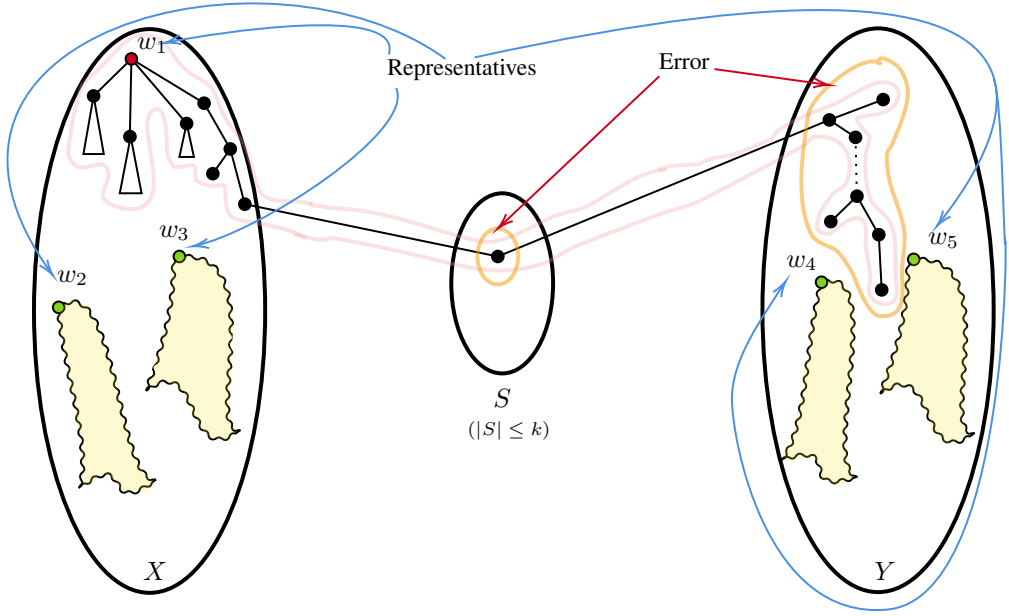

    \centering
    \makebox[\textwidth][c]{

\tikzset{every picture/.style={line width=0.75pt}} %set default line width to 0.75pt        


}
    \caption{Error, i.e. part of a subtree not on the side of its representative}
    \label{fig:error}
\end{figure}
Now, one might think, why not just check all the possibilities of the representatives going to $X, Y,$ or $S$? The reason that this simple idea does not  work is that if a representative goes to the separator, its entire substree of arbitrary size can go anywhere and we do not have any control over it. Reed handles this problem by deciding whether any representative is going to the separator, at the very beginning of the algorithm. If so, he just places such a representative (namely $v$) into $S$ (and not its subtree) and starts forming a new group of representatives by running a new DFS on $G - \{v\}$. So, the other representatives might change. Also, since one vertex has been placed into the separator, now $k \leftarrow k-1$. However, if none of the representatives goes to the separator, he branches on placing them left ($X$) or right ($Y$). This is the high-level idea of Reed's algorithm.

Let's start with presenting and reviewing some definitions.

%\section{Analysis of Reed's Algorithm}\label{sec:analysisReed}

%%We use the standard definitions of treewidth and tree decomposition. Following Reed, we simplify expressions by reserving the letter $k$ for the treewidth$+1$. 
%
%%  nice tree decomposition are not needed.
%%We use the standard definitons of treewidth, tree decomposition, and nice tree decomposition. Of the latter, we only use the simple property that adjacent bags differ in only one vertex. Following Reed, we simplify expressions by reserving the letter $k$ for the treewidth$+1$. 
%
%% In 1992, Reed gave an elegant algorithm \cite{reed1992finding} to either construct a tree decomposition of width at most $7k$ of a given graph $G$, or declare that the treewidth is greater than $k$ and outputs the subgraph which is the bottleneck (no separator of size $\leq k$). 

%Before analyzing Reed's algorithm \cite{reed1992finding}, we introduce and review some definitions.

\subsection{Centroids and Separators}
For an undirected graph $G = (V, E)$ and a subset $W$ of the vertices, $G[W]$ is the subgraph induced by $W$. For the sake of simplicity throughout this paper, let $G - W$ be $G[V\setminus W]$ and $G - v$ be $G - \{v\}$ for any $W \subseteq V(G)$ and any $v \in V(G)$.

Also, in a weighted graph, a non-negative integer weight $w(v)$ is defined for each vertex $v$. For a subset $W$ of the vertices, the weight $w(W)$ is simply the sum of the weights of all vertices in $W$. Furthermore, the total weight or the weight of $G$ is the weight of $V$.

%% Defs 1,2 and L1 removed.     Reinserted
\begin{definition}
A centroid of a weighted tree $T$ is a node $x$ such that none of the trees in the forest $T-x$ has more than half the total weight.
\end{definition}
\begin{definition}
A tree decomposition is called a \emph{good tree decomposition} if the adjacent bags differ in at most one vertex.
\end{definition}
For good tree decompositions we choose a stronger version of centroid.

\begin{lemma}
	Any tree decomposition $\mathcal{T}$ with width $k$ can be converted in linear time to a good tree decomposition $\mathcal{T}^*$ of the same graph without increasing the width. 
\end{lemma}
\begin{proof}
    Bodlaender \cite{bodlaender1996linear} has defined the more restrictive notion of a smooth tree decomposition and shown that it can be computed in linear time. However, we include the derivation for our notion to make it self-contained\footnote{Alternatively, one can refer to the derivation of a similar result for nice tree decomposition~\cite{kloks1994treewidth}.}.
    
	Let $x$ and $y$ be two adjacent nodes of $\mathcal{T}$ with associated bags $B_x$, and $B_y$, respectively.
	
	Let $B_x = \{v_1, \dots, v_t\}$, and $B_y = \{u_1, \dots, u_s\}$. Notice that $s,t \leq k$.
	
	Now, let $V^- = B_x\setminus B_y$, and $U^+ = B_y \setminus B_x$. Let $D = V^- \cup U^+ = \{d_1, \dots, d_{|V^-|}, d'_1, \dots, d'_{|U^+|}\}$. 
	
	We just add $|D|$ intermediate nodes (namely $x_1, x_2, \dots, x_{|V^-|}, y_1, y_2, \dots, y_{|U^+|}$) between $x$ and $y$ such that we start by deleting vertices of $V^-$ one by one and then adding vertices of $U^+$ one at a time. 
	
	Set 
	\begin{align}
		\begin{cases}
			B_{x_1} = B_x \setminus \{d_1\} &\\
			B_{x_{i}} = B_{x_{i-1}} \setminus \{d_i\} & \forall i \in \{2, \dots, |V^-|\}\\
			B_{y_{1}} = B_{x_{|V^-|}} \cup \{d'_1\} &\\
			B_{y_{i}} = B_{y_{i-1}} \cup \{d'_i\} & \forall i \in \{2, \dots, |U^+|\}.\\
		\end{cases}
	\end{align} 
\end{proof}
\begin{definition}
A strong centroid of a good tree decomposition $\tau$ of a graph $G = (V, E)$ with respect to $W \subseteq V$ is a node $x$ of $\tau$ such that none of the connected components of $G - B_x$ contains more than $\frac{1}{2}|W\setminus B_x|$ vertices of $W$.
\end{definition}
The following lemma shows the existence of a strong centroid for any given subset $W$ of $V$.
\begin{lemma}\label{lem:cent}
For every good tree decomposition $(\mathcal{T}, \{B_x: x\in V_{\mathcal{T}}\})$ of a graph $G = (V, E)$ and every subset $W \subseteq V$, there exist a strong centroid with respect to $W$\footnote{Flum-Grohe~\cite{Flum:2006:PCT:1121738} use the property of a standard centroid for tree decompositions in Lemma 11.16 page 267. They talk about balanced $W$-separators and do not use the term centroid. They only use the standard notion of a centroid, showing  that no connected component contains more than $\frac{|W|}{2}$ of the vertices. Strong centroids show that no connected component has size more than $\frac{|W\setminus S|}{2}$.}.
\end{lemma}

\begin{proof} %[Lemma \ref{lem:cent}] 
	If a node $x$ is not a strong centroid with respect to $W$, then let $C_x$ be the set of vertices in the unique connected component of $G-B_x$ containing more than $\frac{1}{2}|W \setminus B_x|$ vertices of $X$. In the forest obtained from the tree $\mathcal{T}$ by removing $x$, there is a tree $\mathcal{T}_x$ with the property that the union of all bags in $\mathcal{T}_x$ contains all the vertices of $C_x$.
	
	Now, we define a set $F$ of directed tree edges by $(x,y) \in F$ if all the following conditions hold:
	\begin{itemize}
		\item $x$ is not a strong centroid
		\item $y$ is a neighbor of $x$ in $\mathcal{T}$,
		\item $y$ is a node in $\mathcal{T}_x$.
	\end{itemize}
	Now we show that there is a node $x$ with out-degree 0 in $(V_{\mathcal{T}}, F)$. Such an $x$ is a centroid, and we are done. Otherwise, $F$ contains $(x,y)$ and $(y,x)$ for some $x,y \in V_{\mathcal{T}}$. W.l.o.g., $B_y = B_x\cup\{v\}$ for some $v\in V\setminus B_x$. Note that $\mathcal{T}_x$ and $\mathcal{T}_y$ are disjoint. Furthermore, any vertex that is in a bag of $\mathcal{T}_x$ and in a bag of $\mathcal{T}_y$ is also in $B_x$ and $B_y$. Thus also $C_x$ and $C_y$ are disjoint.
	
	Furthermore, $W \cap C_x \subseteq W\setminus B_x$ and $W \cap C_y \subseteq W\setminus B_y \subseteq W\setminus B_x$. As $W \cap C_x$ and $W\cap C_y$ are disjoint, they cannot both have more than $\frac{1}{2}|W\setminus B_x|$ vertices, which is a contradiction.
	
	Hence, there exists a node $x$ which is a strong centroid.
\end{proof}

%%end of newly inserted

We use the definitions of balanced $W$-separator and weakly balanced $W$ separation from the book of Flum and Grohe \cite{Flum:2006:PCT:1121738}.

\begin{definition} Let $G = (V,E)$ be a graph and $W \subseteq V$. A \emph{balanced $W$-separator} is a set $S \subseteq V$ such that every connected component of $G-S$ has at most $\frac{1}{2}|W|$ vertices.
\end{definition}

\begin{lemma}\label{Lem:ResBalSep}{\rm \cite[Lemma 11.16]{Flum:2006:PCT:1121738}}
Let $G = (V,E)$ be a graph of treewidth at most $k-1$ and $W \subseteq V$. Then there exists a balanced $W$-separator of $G$ of size at most~$k$.
\end{lemma}

We say that a separator $S$ separates $X \subseteq V$ from $Y \subseteq V$ if $C \cap X = \emptyset$ or $C \cap Y = \emptyset$ for every connected component $C$ of $G - S$.
\begin{definition}
Let $G = (V,E)$ be a graph and $W \subseteq V$. A \emph{weakly balanced separation} of $W$ is a triple $(X, S, Y )$, where $X, Y \subseteq W$, $S \subseteq V$ are pairwise disjoint sets such that:
\begin{itemize}
\item $|S|\leq k$
\item $W = X \cup (S \cap W) \cup Y$
\item $S$ separates $X$ from $Y$
\item $0 < |X|, |Y| \leq \frac{2}{3}|W|$.
\end{itemize}
\end{definition}

\begin{lemma}\label{lem:exs_weak}{\rm \cite[Lemma 11.19]{Flum:2006:PCT:1121738}}
For $k \geq 3$, let $G = (V, E)$ be a graph of treewidth at most $k-1$ and $W \subseteq V$ with $|W| \geq 2k+1$. Then there exists a weakly balanced separation of $W$ of size at most $k$.
\end{lemma}

Even though Lemma~\ref{lem:exs_weak} is sufficient for us, one can make it stronger such that it holds for $|W|\geq k+1$.

\begin{theorem}{\rm \cite[Corollary 11.22]{Flum:2006:PCT:1121738}}
For a graph of treewidth at most $k-1$ with a given set $W \subseteq V$ of size $|W| = 3k-2$, a weakly balanced separation of $W$ of size $\leq k$ can be found in time $O(3^{3k} k^2 n)$.
\end{theorem}

\subsection{Algorithm to Find a Weakly Balanced Separation}\label{subsec:algs}
Separation$(G, k)$ is the main part of Reed's algorithm. It finds a separator of size at most $k$ in $G$ using the procedures Split$(G, X, Y, k)$ and DFS-Trees$(G, k)$. We explain each of these procedures.
\vspace*{-2ex}
\subsubsection{Split$(G, X, Y, k)$}
For $X$, $Y$ disjoint subsets of $V$, Split$(G, X, Y, k)$ finds a separator $S$ of size at most $k$ in $G$ which is strictly between $X$ and $Y$. Split reports failure if no such separator exits (described in Lemma 11.20 of \cite{Flum:2006:PCT:1121738}).
\vspace*{-2ex}
\subsubsection{DFS-Trees($G, k$)}
DFS-Trees$(G, k)$ (Algorithm 1 using Algorithm 2)
computes a DFS tree and partitions it into smaller DFS trees with the following properties.
\begin{itemize}
\item the size (number of vertices) of the smaller trees is at least $s = n/24k$, and
\item all subtrees rooted at children of the roots of the trees in the partition have size less than $n/24k$.
\end{itemize}
DFS-Trees$(G, k)$ collects the set $W'$ consisting of all the roots of the trees in the partition.
These roots are representatives of the vertices in their small DFS tree. 
Therefore, the weight $w[v]$ for $v \in W'$ is the number of vertices in the small tree with root~$v$.
\vspace*{-2ex}

\subsubsection{Separation($G, k$)}
Separation($G, k$) is the recursive procedure that splits according to the number of vertices (Algorithm \ref{alg:separation}). Note that when any vertex $v$ is placed into the separator $S$, then the procedure Separation removes that vertex $v$ from the graph and starts from scratch. The idea is that when we place a root of a small tree (a representative) left or right, then we want to put the whole small tree there. But when a representative is placed into the separator, then its tree does not go there. At this point a new collection of trees is formed.

\begin{algorithm}\label{alg:DFS}
	\DontPrintSemicolon
	\SetAlgoLined
	\KwResult{Roots of DFS-Trees  (representatives) with sizes of their strict subtrees $< |V| / (24 k)$}
	\textbf{Procedure DFS-Trees($G, k$)} \tcp*{$G$ is a connected graph.}
	$_{^*}s = \frac{|V|}{24k}$ \tcp*{$s:$ the size bound for splitting off a small tree.}
	$_{^*}W' = \emptyset$\;
	$_{^*}$\For{all $v \in V$}{
		$_{^*}$color$[v] = $ WHITE\; 
	}
	%$_{^*}$count $= 0$\;
	$_{^*}$Pick any vertex $u$ of $G$.\;
	$_{^*}$count = DFS-visit($G, u$)\;
	$_{^*}w[x] = w[x] + count$, where $x$ is the last representative added to $W'$\;
	%$_{^*}$Add count to $w[v]$, where $v$ is the vertex last included in $W'$.\;
	$_{^*}$\Return{$(W'$,  $w[v]$ for all $v \in W')$}\;
	%$_{^*}$\textbf{output }($W'$,  $w[v]$ for all $v \in W'$)\;
	\textbf{End Procedure}
	\caption{Construct Small DFS-Trees}
	\label{alg:DFS}
\end{algorithm}

\begin{algorithm}
	\DontPrintSemicolon
	\SetAlgoLined
	\textbf{Procedure DFS-Visit($G, u$)} \\
	$_{^*}$ color[$u$] = GRAY\;
	$_{^*}$ count = 1\;
	$_{^*}$\For{all $v$ adjacent to $u$}{
		$_{^*}$\If{color$[v] == $ WHITE \tcp*{The white vertex $v$ is discovered now.}}{ 
			$_{^*}$ $count = count + \text{DFS-Visit($G, v$})$\;
			%$_{^*}$\If{child\_count $<$ s}{
			%	count += child\_count\;
			%}
		} 
	}
	$_{^*}$\If{$count \geq s$}{
		$_{^*}$ $W' = W' \cup \{u\}$\;
		$_{^*}$ $w[u] = count$\;
		$_{^*}$ $count=0$\;
	}
	\Return{count}
	%$_{^*}$color[$u$] = BLACK\;
	%$_{^*}$count = count+1\; 
	%$_{^*}$\If{count $\geq s$ \tcp*{$u$ is a root of a small tree.}}{ 
		%$_{^*}W' = W' \cup \{u\}$\;
		%$_{^*}w[u] = count$\;
		%$_{^*}count = 0$\; 
	%}
	
	\textbf{End Procedure}
	\caption{Main recursive procedure of DFS-Trees}
\end{algorithm}

\begin{algorithm}
	\DontPrintSemicolon
	\SetAlgoLined
	\KwResult{A  weakly balanced separation $(X, S, Y)$ of $G$ of size $\leq k$}
	\textbf{Procedure SEPARATION($G, k$)} \\
	$_{^*}$\If{$G$ is not a connected graph}{
	$_{^*}$Let $C_1, \dots, C_t$ be the connected components of $G$, and w.l.o.g., assume $C_1$ is the largest component.\\
	$_{^*}$\If{$|C_1|< \frac{3}{4}|V|$}{
	$_{^*}$Let $L = C_1$ and $i=2$.\\
	$_{^*}$\While{$|L| < \frac{1}{4}|V|$}{
	$_{^*} L = L\cup C_i$.\\
	$_{^*} i = i+1$ 
	}
	$_{^*}$\Return{$(L, \emptyset, V\setminus L)$}
	}
	$_{^*}$Let $(X',S',Y') = $SEPARATION$(G[C_1], k)$. W.l.o.g., assume $X'$ is the one with lower weight.\\
	$_{^*}$\Return{$(X'\cup (\bigcup\limits_{i=2}^t C_i)), S', Y')$}
	}
	$_{^*}$\If{$k >0$}{$_{^*} (W', w[v]: \forall v \in W') =$ DFS-Trees($G, k$)\;}
	$_{^*}$\For{all $v \in W'$ \tcp*{Here $v$ is placed into separator $S$.}}{
		$_{^*}(X, S, Y) = $ SEPARATION($G - v, k - 1$)\;
		$_{^*}$\If{$\neg$failure}{
			$_{^*}$\Return{$(X, S\cup\{v\}, Y)$}} 
	}
	\tcp*{The set of vertices $W'$ is partitioned into $X\subseteq L$ and $Y \subseteq R = W' \setminus L$.}
	$_{^*}$\For{all $X\subseteq W'$ \tcp*{Here no vertex is put into $S$.}}{
		$_{^*}$\If{$(\frac{1}{3} - \frac{1}{24})|V| \leq w(X) \leq (\frac{2}{3} + \frac{1}{24})|V|$}{
			$_{^*}$Split($G, X, W'\setminus X, k$)\;
			$_{^*}$\If{$\neg$failure}{
				$_{^*}$\Return{($X, S, Y$)}
			}
		}
	}
	\Return{failure}\;
	\textbf{End Procedure}
	\caption{Main recursive procedure in Reed's algorithm}\label{alg:separation}
\end{algorithm}

\nopagebreak[4]

\subsection{The Correctness of Reed's Algorithm}\label{subse:correct_Reed}
If the treewidth is at most $k-1$, then there is a good tree decomposition of $G$ of width $k-1$. Let $x$ be a centroid in it. The connected components of $G[V \setminus B_x]$ can be partitioned into 2 parts $L$ and $R$, such that no part has more than $\frac{2}{3} |V|$ vertices. 

First, we prove the correctness of the algorithm for the case that $G$ is a connected graph (lines 16-25 of Algorithm~\ref{alg:separation}). Later, we describe the case that $G$ is not connected.

Note that for the correctness proof, we do not have to find this tree decomposition. It is sufficient to know that it exists. We can assume, that we have fixed such a tree decomposition, a centroid $x$ and the sets $L$ and $R$.

Every set $W'\subseteq V$ is partitioned into parts in $L$, the separator $S=B_x$, and $R$. 

One of the many branches of the procedure Separation$(G, k)$ working with a set $W'$, tries this partition of it, and succeeds, unless a previously taken branch has already succeeded. 
First, the procedure decides which part of $W'$ goes into $S$, one vertex $v$ at a time. This vertex $v$ is removed from $G$, but otherwise, we still consider the same tree decomposition. $|B_x|$ has now decreased by 1, as $v$ is removed from it.

We then consider the case that none of the remaining vertices in $W'$ are in the separator. The procedure decides which part $X \subseteq W'$ goes into $L$ (line 26) of Algorithm~\ref{alg:separation}. Now the weight of each part is at most $(\frac{2}{3} + \frac{1}{24}) n$ as at most $k$ small subtrees can have some of their vertices on the wrong side. And this is at most $k$ times less than $\frac{n}{24k}$ vertices.

On the branch of the procedure Separation$(G, k)$ which selects this partition of $W'$, there is the separator $B_x$ of size at most $k$ between $X$ and $Y$. Our algorithm cannot guarantee to find this separator $B_x$, but it will find some separator $S$ of size at most $k$ between $X$ and $Y$. Again up to $\frac{1}{24} n$ vertices can be on the opposite side of their representatives. Now the larger side can contain at most $(\frac{2}{3} + 2 \cdot \frac{1}{24}) n = \frac{3}{4} n$ vertices.
Thus we have a somewhat balanced partition (a constant fraction on each side).

Note that if at any point $G$ becomes disconnected in Algorithm~\ref{alg:separation}, it only works to our advantage, and we handle it separately in lines 2-15. Let $C_1, \dots, C_t$ be the connected components of $G$, and w.l.o.g., assume $C_1$ has the highest volume. In this case, there are two possibilities:
\begin{itemize}
    \item First, we consider the case that $C_1$ has volume (actual size) $\leq \frac{3}{4}n$. Then, if $|C_1|$ is already $\geq \frac{1}{4} n$ we are done. We return $C_1$ as the L.H.S. and the remaining components as the R.H.S. However, if $|C_1| < \frac{1}{4} n$, we add other components to $C_1$ until its volume passes $\frac{1}{4} n$ for the first time. Note that since $|C_1|$ has maximum volume, all other components have volume $\leq \frac{1}{4} n$, and adding them one at a time will not cause an issue.
    \item In the second case, $C_1$ has volume more than $\frac{3}{4}n$. Then, we run \textbf{SEPARATION}$(G[C_1], k)$ to get $(X', S', Y')$. This is a separation for $G[C_1]$. W.l.o.g., assume $X'$ does not have a higher weight than $Y'$. Then, we put all the vertices represented by $X'$ along with all other components ($X' \cup (\bigcup\limits_{i=2}^{t}C_i$)) on the L.H.S. 
    Note that:
    \begin{align*}
        \text{Vol}(X'\cup (\bigcup\limits_{i=2}^{t}C_i)) = \text{Vol}(X') + n - |C_1|.
    \end{align*}
    We know that $(\frac{1}{3}-\frac{1}{24})|C_1|\leq Vol(X') \leq (\frac{1}{2}+ \frac{1}{24})|C_1|$. Hence,
    \begin{align*}
        \text{Vol}(X'\cup (\bigcup\limits_{i=2}^{t}C_i))&\leq (\frac{1}{2}+\frac{1}{24})|C_1| + n - |C_1| = n - \frac{13}{24}|C_1| 
        \\&\leq \frac{19}{32}n < \frac{3}{4}n. 
    \end{align*}
    and
    \begin{align*}
        \text{Vol}(X'\cup (\bigcup\limits_{i=2}^{t}C_i)&\geq (\frac{1}{3}-\frac{1}{24})|C_1| + n - |C_1| = n - \frac{17}{24}|C_1|
        \\&\geq n - \frac{17}{24}n = \frac{7}{24}n >  \frac{1}{4}n. 
    \end{align*}
\end{itemize}

Furthermore,
\begin{align*}
    \text{Vol}(Y')\leq \frac{3}{4}|C_1| \leq \frac{3}{4}n,
\end{align*}
and
\begin{align*}
    \text{Vol}(Y') \geq (\frac{1}{2}-\frac{1}{12})|C_1|\geq \frac{5}{12}\cdot \frac{3}{4}n = \frac{5}{16}n > \frac{1}{4}n.
\end{align*}
So, in the case that $G$ is not connected, we still find a somewhat balanced separator.

We assume the $\mathcal{O}\left(3^{3k} k^2 n^2\right)$ 4-approximation algorithm of Robertson-Seymour \cite{robertson1995graph,reed_1997} (see Proposition 11.14 of Flum and Grohe \cite{Flum:2006:PCT:1121738}) is known. It handles a set $W \subseteq V$ of size $3k-2$. Working with $(G, W)$, it finds a separator $S$ of size at most $k$ to split $G-S$ into two parts $L$ and $R$ with both, $|W \cap X|$ and $|W \cap Y|$ at most $\frac{2}{3}|W|$. A tree node with bag $W \cup S$ is formed. Two recursive calls continue with $(G[L \cup S], (W \cap L) \cup S)$ and $(G[R \cup S], (W \cap R) \cup S)$, respectively. 

Reed's algorithm does the same steps to handle $W$. However, in order to decrease the dependence on $n$ of the running time from $\mathcal{O}\left(n^2\right)$ to $\mathcal{O}(n \log n)$, Reed intersperses these balanced partitions of $W$ with balanced partitions of $V$ (Algorithm~\ref{alg:separation}). In each case, it would be desirable that $V$ and $W$ simultaneously split in a balanced way. However, during the traditional splitting of $W$, the graph might be split very unbalanced, and during the new splitting of $V$, the set $W$ might not be split at all.

Reed's algorithm can alternate between splitting $W$ as in the $\mathcal{O}\left(3^{3k} k^2 n^2\right)$ algorithm and splitting $V$.
Now $W$ can be of size at most $6k$. On each side, we have at most $(\frac{2}{3} \cdot 6k) = 4k$. Splitting by $W$ as well as splitting by $V$ adds $k$ to the new $W$. Thus, we are back at $6k$. The constructed tree decomposition has then width at most $7k$. But we show that this can be improved to a 5-approximation algorithm. We do not need to alternate between splitting $W$ and $V$. Splitting $V$ is a costly procedure. We can do it only after every $\log_{\frac{3}{2}}{k}$ steps and we still spend only $\mathcal{O}(f(k) n\log{n})$ time.

We start with $W$ of size at most $4k$ ($3k$ and $k_{excess}$ as excess). Initially, $k_{excess} = k$. Each time we split $W$, we get 
\[|W| \leq \frac{2}{3}\cdot 3k + \underbrace{k}_{\text{adding separator}} + \frac{2}{3}k_{excess} = 3k +\frac{2}{3} k_{excess},\]
and then update $k_{excess} \leftarrow \lfloor \frac{2}{3}k_{excess}\rfloor$. The excess drops by a factor of $\frac{2}{3}$. After $\log_{\frac{3}{2}}{k}$ steps, the excess becomes zero and then we can split by $V$. At this point, $|W|$ could increase to $4k$ again ($3k$ was the size of $W$ before this step, and when we split by $V$, we have to include the separator as well). We end up with $|W| \leq 4k$ and we add the separator to the root bag, which means the largest bag has size at most $5k$. Therefore, it is a 5-approximation algorithm. Reed mentions $5k$ in his paper, but he does not mention the frequency of the two operations. We don't know whether he had the same modification in mind. Simple alternation between the two operations only achieves $7k$. 
%However, if splitting by $V$ does not happen very often, we can achieve $5k$.

\subsection{Running Time of Reed's Algorithm}\label{sec:time}
Let $T(n,k)$ be the running time of the procedure $SEPARATION(G,k)$ for $G = (V,E)$ and $n = |V|$. Let $n'$ and $k'$ be the current bound on the graph size and current separator capacity. Initially $n' = n$ and $k' = k$. We have the following recurrence for Reed's algorithm.
\begin{equation}
	T(n',k') \leq 24k'T(n'-1, k'-1)+ 2^{24k'}  \underbrace{c(k'+1)kn'}_\text{flow algorithm},
\end{equation}
for some $c>0$.
It is difficult to obtain a good solution, but by induction on $k'$ we get the following loose upper bound.
\begin{equation}
T(n',k') \leq 3 c 2^{24k'} k'! \, kn.
\end{equation}
For $k'=0$, this bound is valid. For $k'\geq 1$, we have:
\begin{equation*} %%%***
\begin{split}
T(n',k') &\leq 24k'T(n'-1, k'-1)+ c2^{24k'}(k'+1)kn'\\
&\leq 3 c 24 k' 2^{24(k'-1)}(k'-1)! \,kn + c2^{24k'}(k'+1)kn'  \qquad \text{by induction hypothesis}\\
&= c2^{24k'}kn(\frac{3\cdot 24k'!}{2^{24}} + k'+1) \leq 3c2^{24k'}knk'!\,(\frac{24}{2^{24}} + \frac{k'+1}{3\, k'!})\\
&\leq 3c2^{24k'}k'!\,kn.
\end{split}
\end{equation*}
Even though, this is not a tight bound, we have $T(n,k) \geq c' 24^k k! (n-k)$, which is $2^{\Omega (k \log k)}n$.
\begin{eqnarray*}
T(n,0) & \geq & c' n \\
T(n,k) & \geq & 24k T(n-1,k-1) \\
	& \geq & c' 24k 24^{k-1} (k-1)!  \, (n-k) \text{ by inductive hypothesis} \\
	& \geq & c' 24^k k! \,  (n-k)
\end{eqnarray*}
%%% **** Here the asymptotic notation is a bit misleading from a practical point of view, as $2^{24k} = o(k!)$, even though $k!$ is reasonable for small $k$, while $2^{24k}$ is not.

\section{Our Improved Algorithm}\label{sec:ours}
In this section, we discuss how to improve Reed's algorithm. The dependence on $k$ in the running time of Reed's algorithm is huge. We decrease this factor significantly to make the algorithm more applicable. We introduce two main modifications. First, we work with a larger cut-off threshold than Reed's $\frac{|V|}{24k}$. Such an improvement can be achieved by replacing the arbitrary $3/4$ bound by $1- \epsilon$. But even more is possible by arguing about the weights of connected components instead of the weights of the parts of a bipartition.

The second improvement is to avoid branching on whether there is a representative going into the separator or not. Reed branches on these two cases at the beginning, while we branch 3-fold for every representative.

Note that the second improvement is not obvious. Reed had a good reason to avoid 3-fold branching. If a representative of tree is put into the separator, then, we lose control of the unbounded set of additional vertices of the tree.

\subsection{Relax the balancing requirement}\label{sec:divide}
Reed's argument starts with a weakly balanced separation by volume (i.e. according to $V$) that is known to exist. The larger side has at most $2/3$ of the volume, but it might have up to $2/3 + 1/24$ of the weight. The algorithmic split by this weight partition might pick a set with another $1/24$ fraction more volume. Thus the worst kind of volume split found is now $3/4$ to $1/4$. Recall that these differences are bounded this way for the following reason. When the weight carrying root of a tree is on one side, some of its small subtrees rooted at the children can be partially on the other side. But the separator prevents more than $k$ small subtrees to have any part on a different side than the root, and each small subtree contains less than $\frac{|V|}{24k}$ vertices. 
Instead of $3/4$, one can choose any number strictly between $2/3$ and $1$. If $1-\epsilon$ is chosen, then the constant 24 is replaced by $\frac{1}{((1-\epsilon) - 2/3)/2} = \frac{6}{1 - 3 \epsilon} \leq 6 + 24 \epsilon$ for $\epsilon \leq 1/12$.

A better improvement is possible by a modification of the analysis. 
If $G$ is connected, then
we know that a balanced $V$-separator $S'$ exists. It is a separator by volume. 
We fix such an $S'$ for the analysis. The algorithm does not have to find it.
$S'$ defines a set of connected components $G_1, \dots, G_{t}$ of $G-S'$ with vertex sets $V_1, \dots, V_{t}$.
%Now, for each branch of the algorithm, we consider the set of vertices $W''$ that it puts into the separator.
%Assume, this branch selects the set $W'$ of representatives.
We focus on the set $\cal B$ of those branches of the algorithm which put any vertex $v \in W'$ into the separator if and only if $v \in S'$.
For each branch of $\cal B$, the weight of $G_i$ is the weight of $W' \cap V_i$.
We choose a larger value of $s$, namely
$s = \lfloor(1/4 - \epsilon/2)\frac{|V|}{k}\rfloor$.
Because $|V_i| \leq \frac{|V|}{2}$, the weight of each connected component $G_i$ is at most $(3/4 - \epsilon/2)|V|$.
We group $V_1, \dots, V_{t}$ into $L'$ and $R'$ such that neither $L'$ nor $R'$ have weight more than $(3/4 - \epsilon/2)|V|$. 

Now, there exists exactly one branch of $\cal B$ that places a vertex $v \in W'$ left if and only if $v \in L'$. Otherwise such a vertex $v \in W'$ is placed right on this branch. 
Again, note that the algorithm does not know $L'$ and $R'$. However, as it tries all possible placements of $W'$, one branch is good. 

The algorithm finds a separator $S$ that allows the same 2-partition of $W' \setminus S'$ as $S'$. The separator $S$ found by the algorithm might be different from $S'$. However, as $S'$ and $S$ agree on $W'$, the volume ratio is not worse than $(1-\epsilon)$ to $\epsilon$.

With this improvement, the constant 24 of Reed's algorithm is replaced by
 \[\frac{1}{((1-\epsilon) - 1/2)/2} = \frac{4}{1 - 2 \epsilon} \leq 4 + 12 \epsilon \text{ for $\epsilon \leq 1/6$.}\]

\subsection{Main Improvement}
The other improvement is to allow the representatives (the roots of the subtrees) to go either, left, right, or into the separator. Once a representative $v$ goes into the separator, we change its weight to 0. We also delete $v$ from $G$ and unmark all of the vertices in its subtree, so that they can be searched again. Then, we continue the DFS from the parent of $v$.  There is one complication here that we have to take care of; what happens if $G-v$ gets disconnected?

Three cases might happen. We cover all the cases and show that they can be handled. Let $x, y,$ and $z$ be three types of children of $v$ with the small subtrees (with size $ < \frac{n}{Ck}$) $\tau_x$, $\tau_y$, and $\tau_z$, rooted respectively (Fig. \ref{fig:fixing}(a)). The cases are:
\begin{itemize}
	\item There is a back edge from a vertex in $\tau_x$ to an ancestor of $v$. If we delete $v$ from the tree, the vertices of $\tau_x$ will still be searched because they are connected to an ancestor of $v$. So, we just need to color them white once more. Hence, we do not need to worry about this case.
	\item There is a vertex $p$ in $\tau_y$ with subtree $\tau_q$ attached to it with $q$ as its root such that $q$ is a representative below $\tau_y$ which has not gone into the separator (this is a bottom-up approach). Even though this case seems to be troublesome, we can fix it. Let $\tau_{p, y}$ be $\tau_y$ rooted at $p$ (dangling from $p$). Make $p$ a child of $q$. Now, the problem has been fixed (Fig. \ref{fig:fixing}(b)). The same reasoning applies when $v$ is a root and its subtree is too small.
	
	Notice that in this case, it is possible that there is a back edge to some vertex higher up as well, but it does not change anything here. We handle it by attaching to the subtree below it. Alternatively, one could handle it like the previous case. Both are fine.
	\item There is no back edge from $\tau_z$ to an ancestor of $v$, and there is no subtree below. So deleting $v$ makes $\tau_z$ disconnected from the entire tree, and this is only for our advantage.
\end{itemize}
Now, the question is how to distinguish between these three cases?

We just need to distinguish between case 2 and the other two cases (together) because case 1 and 3 are handled the same way.

Cases 1 and 3 are actually the simple ones. We just need to color the vertices in their subtrees white and then we are done. So, we have to find a way to recognize case 2. For that, each vertex $u$ in $\tau(y)$ saves a possible representative $q$ with an edge to some $p\in \tau(y)$ such that $p$ is an ancestor of $q$ (defined as $rep\_\,connected$ in the pseudocode). If no such representative exists, simply set its value to nil. Notice that we just need to store one representative connected down below since it is sufficient to handle this case as we just described. Each vertex passes that information to its parent. So, if $rep\_\,connected[y] \neq nil$, we know that the case 2 has happened. Otherwise, case 1 or 3 has happened and we do not need to know which one since they are handled the same way. Now, let's go back the algorithm description.
\tikzset{every picture/.style={line width=0.75pt}} %set default line width to 0.75pt        

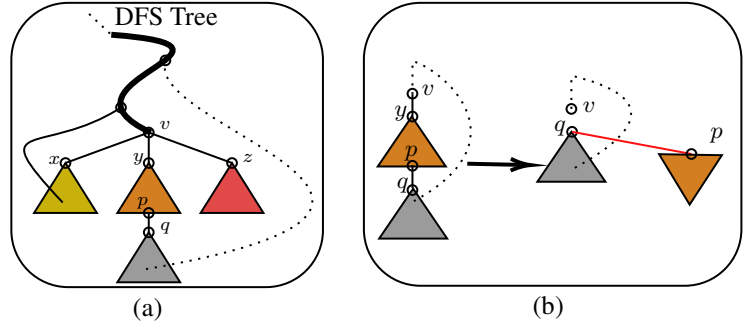
\begin{figure}
	\centering

	\tikzset{every picture/.style={line width=0.75pt}} %set default line width to 0.75pt        

	\begin{tikzpicture}[x=0.75pt,y=0.75pt,yscale=-0.7,xscale=0.7]
		%uncomment if require: \path (0,434); %set diagram left start at 0, and has height of 434
		
		%Shape: Triangle [id:dp4445463156620548] 
		\draw  [fill={rgb, 255:red, 155; green, 155; blue, 155 }  ,fill opacity=1 ] (295,170) -- (319.5,206) -- (270.5,206) -- cycle ;
		%Shape: Circle [id:dp7753968583101003] 
		\draw   (101,129.5) .. controls (101,127.57) and (102.57,126) .. (104.5,126) .. controls (106.43,126) and (108,127.57) .. (108,129.5) .. controls (108,131.43) and (106.43,133) .. (104.5,133) .. controls (102.57,133) and (101,131.43) .. (101,129.5) -- cycle ;
		%Shape: Triangle [id:dp8852514428184302] 
		\draw  [fill={rgb, 255:red, 201; green, 177; blue, 12 }  ,fill opacity=1 ] (44.5,151.5) -- (67.13,187.5) -- (21.88,187.5) -- cycle ;
		%Straight Lines [id:da43610850522852185] 
		\draw    (104.5,129.5) -- (44.5,151.5) ;
		%Straight Lines [id:da2172341023833282] 
		\draw    (104.5,129.5) -- (104.5,151.5) ;
		%Straight Lines [id:da011176509266443091] 
		\draw    (104.5,129.5) -- (164.5,151.5) ;
		%Curve Lines [id:da6265577772407787] 
		\draw [line width=2.25]    (77.5,59) .. controls (197.5,68) and (26.5,93) .. (106,129.5) ;
		%Curve Lines [id:da4523796818486352] 
		\draw    (84.5,111.5) .. controls (28.5,147.5) and (-16.5,117) .. (42.5,180) ;
		%Shape: Circle [id:dp5932487568179552] 
		\draw   (81,111.5) .. controls (81,109.57) and (82.57,108) .. (84.5,108) .. controls (86.43,108) and (88,109.57) .. (88,111.5) .. controls (88,113.43) and (86.43,115) .. (84.5,115) .. controls (82.57,115) and (81,113.43) .. (81,111.5) -- cycle ;
		%Shape: Circle [id:dp7122306014731352] 
		\draw   (113,77.5) .. controls (113,75.57) and (114.57,74) .. (116.5,74) .. controls (118.43,74) and (120,75.57) .. (120,77.5) .. controls (120,79.43) and (118.43,81) .. (116.5,81) .. controls (114.57,81) and (113,79.43) .. (113,77.5) -- cycle ;
		%Straight Lines [id:da10532548813549658] 
		\draw  [dash pattern={on 0.84pt off 2.51pt}]  (61.5,44) -- (77.5,59) ;
		%Shape: Triangle [id:dp12296482354744764] 
		\draw  [fill={rgb, 255:red, 211; green, 126; blue, 33 }  ,fill opacity=1 ] (104.5,151.5) -- (127.13,187.5) -- (81.88,187.5) -- cycle ;
		%Shape: Triangle [id:dp6300979427798594] 
		\draw  [fill={rgb, 255:red, 226; green, 74; blue, 74 }  ,fill opacity=1 ] (164.5,151.5) -- (187.13,187.5) -- (141.88,187.5) -- cycle ;
		%Shape: Circle [id:dp5951870858752337] 
		\draw   (161,151.5) .. controls (161,149.57) and (162.57,148) .. (164.5,148) .. controls (166.43,148) and (168,149.57) .. (168,151.5) .. controls (168,153.43) and (166.43,155) .. (164.5,155) .. controls (162.57,155) and (161,153.43) .. (161,151.5) -- cycle ;
		%Shape: Circle [id:dp5269152819887553] 
		\draw   (41,151.5) .. controls (41,149.57) and (42.57,148) .. (44.5,148) .. controls (46.43,148) and (48,149.57) .. (48,151.5) .. controls (48,153.43) and (46.43,155) .. (44.5,155) .. controls (42.57,155) and (41,153.43) .. (41,151.5) -- cycle ;
		%Shape: Circle [id:dp7528011640722094] 
		\draw   (101,151.5) .. controls (101,149.57) and (102.57,148) .. (104.5,148) .. controls (106.43,148) and (108,149.57) .. (108,151.5) .. controls (108,153.43) and (106.43,155) .. (104.5,155) .. controls (102.57,155) and (101,153.43) .. (101,151.5) -- cycle ;
		%Shape: Triangle [id:dp2399334289222932] 
		\draw  [fill={rgb, 255:red, 155; green, 155; blue, 155 }  ,fill opacity=1 ] (104.5,201.5) -- (127.13,237.5) -- (81.88,237.5) -- cycle ;
		%Shape: Circle [id:dp20501789900235834] 
		\draw   (101,187.5) .. controls (101,185.57) and (102.57,184) .. (104.5,184) .. controls (106.43,184) and (108,185.57) .. (108,187.5) .. controls (108,189.43) and (106.43,191) .. (104.5,191) .. controls (102.57,191) and (101,189.43) .. (101,187.5) -- cycle ;
		%Curve Lines [id:da6670654582931279] 
		\draw  [dash pattern={on 0.84pt off 2.51pt}]  (102.5,228) .. controls (375.5,181) and (93.5,124) .. (117.5,77) ;
		%Rounded Rect [id:dp5931831641980074] 
		\draw   (5.5,77) .. controls (5.5,54.36) and (23.86,36) .. (46.5,36) -- (189,36) .. controls (211.64,36) and (230,54.36) .. (230,77) -- (230,200) .. controls (230,222.64) and (211.64,241) .. (189,241) -- (46.5,241) .. controls (23.86,241) and (5.5,222.64) .. (5.5,200) -- cycle ;
		%Rounded Rect [id:dp10834087740331944] 
		\draw   (259.5,76.8) .. controls (259.5,54.27) and (277.77,36) .. (300.3,36) -- (490.7,36) .. controls (513.23,36) and (531.5,54.27) .. (531.5,76.8) -- (531.5,199.2) .. controls (531.5,221.73) and (513.23,240) .. (490.7,240) -- (300.3,240) .. controls (277.77,240) and (259.5,221.73) .. (259.5,199.2) -- cycle ;
		%Shape: Triangle [id:dp3562868644134749] 
		\draw  [fill={rgb, 255:red, 211; green, 126; blue, 33 }  ,fill opacity=1 ] (294,118) -- (318.5,154) -- (269.5,154) -- cycle ;
		%Shape: Circle [id:dp2661256811867909] 
		\draw   (291,153.5) .. controls (291,151.57) and (292.57,150) .. (294.5,150) .. controls (296.43,150) and (298,151.57) .. (298,153.5) .. controls (298,155.43) and (296.43,157) .. (294.5,157) .. controls (292.57,157) and (291,155.43) .. (291,153.5) -- cycle ;
		%Straight Lines [id:da7031666837137422] 
		\draw    (294.5,153) -- (294.5,169.5) ;
		%Shape: Circle [id:dp949534536904376] 
		\draw   (291,171) .. controls (291,169.07) and (292.57,167.5) .. (294.5,167.5) .. controls (296.43,167.5) and (298,169.07) .. (298,171) .. controls (298,172.93) and (296.43,174.5) .. (294.5,174.5) .. controls (292.57,174.5) and (291,172.93) .. (291,171) -- cycle ;
		%Straight Lines [id:da8801554023665947] 
		\draw [line width=1.5]    (334,152) -- (376.5,152.93) ;
		\draw [shift={(379.5,153)}, rotate = 181.26] [color={rgb, 255:red, 0; green, 0; blue, 0 }  ][line width=1.5]    (14.21,-4.28) .. controls (9.04,-1.82) and (4.3,-0.39) .. (0,0) .. controls (4.3,0.39) and (9.04,1.82) .. (14.21,4.28)   ;
		%Shape: Circle [id:dp5359363087745335] 
		\draw   (291,118) .. controls (291,116.07) and (292.57,114.5) .. (294.5,114.5) .. controls (296.43,114.5) and (298,116.07) .. (298,118) .. controls (298,119.93) and (296.43,121.5) .. (294.5,121.5) .. controls (292.57,121.5) and (291,119.93) .. (291,118) -- cycle ;
		%Straight Lines [id:da8896713907078568] 
		\draw    (294.5,102) -- (294.5,118.5) ;
		%Straight Lines [id:da7735145897452627] 
		\draw  [dash pattern={on 0.84pt off 2.51pt}]  (294.5,102) -- (294.5,78) ;
		%Shape: Circle [id:dp6949010585114457] 
		\draw   (291,101.5) .. controls (291,99.57) and (292.57,98) .. (294.5,98) .. controls (296.43,98) and (298,99.57) .. (298,101.5) .. controls (298,103.43) and (296.43,105) .. (294.5,105) .. controls (292.57,105) and (291,103.43) .. (291,101.5) -- cycle ;
		%Shape: Triangle [id:dp2997668065904395] 
		\draw  [fill={rgb, 255:red, 155; green, 155; blue, 155 }  ,fill opacity=1 ] (409,129) -- (433.5,165) -- (384.5,165) -- cycle ;
		%Shape: Circle [id:dp2816099772742613] 
		\draw   (406,129) .. controls (406,127.07) and (407.57,125.5) .. (409.5,125.5) .. controls (411.43,125.5) and (413,127.07) .. (413,129) .. controls (413,130.93) and (411.43,132.5) .. (409.5,132.5) .. controls (407.57,132.5) and (406,130.93) .. (406,129) -- cycle ;
		%Shape: Circle [id:dp07290137839251143] 
		\draw   (405.5,112.5) .. controls (405.5,110.57) and (407.07,109) .. (409,109) .. controls (410.93,109) and (412.5,110.57) .. (412.5,112.5) .. controls (412.5,114.43) and (410.93,116) .. (409,116) .. controls (407.07,116) and (405.5,114.43) .. (405.5,112.5) -- cycle ;
		%Straight Lines [id:da6859015175890444] 
		\draw  [dash pattern={on 0.84pt off 2.51pt}]  (409,112.5) -- (409,88.5) ;
		%Curve Lines [id:da7513678657841008] 
		\draw  [dash pattern={on 0.84pt off 2.51pt}]  (295.5,179) .. controls (377.5,134) and (314.5,83) .. (294.5,78) ;
		%Curve Lines [id:da5293643364180203] 
		\draw  [dash pattern={on 0.84pt off 2.51pt}]  (411.5,155) .. controls (493.5,110) and (429,93.5) .. (409,88.5) ;
		%Straight Lines [id:da43110563303546234] 
		\draw [color={rgb, 255:red, 254; green, 14; blue, 14 }  ,draw opacity=1 ]   (409.5,129) -- (495.5,145) ;
		%Shape: Triangle [id:dp0340417481167683] 
		\draw  [fill={rgb, 255:red, 211; green, 126; blue, 33 }  ,fill opacity=1 ] (494.2,181.5) -- (472.17,145.13) -- (517.42,145.87) -- cycle ;
		%Shape: Circle [id:dp5814413960095024] 
		\draw   (492,145) .. controls (492,143.07) and (493.57,141.5) .. (495.5,141.5) .. controls (497.43,141.5) and (499,143.07) .. (499,145) .. controls (499,146.93) and (497.43,148.5) .. (495.5,148.5) .. controls (493.57,148.5) and (492,146.93) .. (492,145) -- cycle ;
		%Shape: Circle [id:dp6770239105655684] 
		\draw   (101,201.5) .. controls (101,199.57) and (102.57,198) .. (104.5,198) .. controls (106.43,198) and (108,199.57) .. (108,201.5) .. controls (108,203.43) and (106.43,205) .. (104.5,205) .. controls (102.57,205) and (101,203.43) .. (101,201.5) -- cycle ;
		%Straight Lines [id:da04297556846047024] 
		\draw    (104.5,187.5) -- (104.5,201.5) ;

		% Text Node
		\draw (110,119) node [anchor=north west][inner sep=0.75pt]  [font=\scriptsize] [align=left] {$\displaystyle v$};
		% Text Node
		\draw (30,144) node [anchor=north west][inner sep=0.75pt]  [font=\scriptsize] [align=left] {$\displaystyle x$};
		% Text Node
		\draw (91,142) node [anchor=north west][inner sep=0.75pt]  [font=\scriptsize] [align=left] {$\displaystyle y$};
		% Text Node
		\draw (170,143) node [anchor=north west][inner sep=0.75pt]  [font=\scriptsize] [align=left] {$\displaystyle z$};
		% Text Node
		\draw (110.5,192) node [anchor=north west][inner sep=0.75pt]  [font=\scriptsize] [align=left] {$\displaystyle q$};
		% Text Node
		\draw (78,37) node [anchor=north west][inner sep=0.75pt]   [align=left] {DFS Tree};
		% Text Node
		\draw (287,137) node [anchor=north west][inner sep=0.75pt]  [font=\footnotesize] [align=left] {$\displaystyle p$};
		% Text Node
		\draw (279.5,108) node [anchor=north west][inner sep=0.75pt]  [font=\footnotesize] [align=left] {$\displaystyle y$};
		% Text Node
		\draw (91,247) node [anchor=north west][inner sep=0.75pt]   [align=left] {(a)};
		% Text Node
		\draw (379,245) node [anchor=north west][inner sep=0.75pt]   [align=left] {(b)};
		% Text Node
		\draw (281,160) node [anchor=north west][inner sep=0.75pt]  [font=\footnotesize] [align=left] {$\displaystyle q$};
		% Text Node
		\draw (299,93) node [anchor=north west][inner sep=0.75pt]  [font=\footnotesize] [align=left] {$\displaystyle v\ $};
		% Text Node
		\draw (394.5,119) node [anchor=north west][inner sep=0.75pt]  [font=\footnotesize] [align=left] {$\displaystyle q$};
		% Text Node
		\draw (415,105) node [anchor=north west][inner sep=0.75pt]  [font=\footnotesize] [align=left] {$\displaystyle v\ $};
		% Text Node
		\draw (507,124) node [anchor=north west][inner sep=0.75pt]  [font=\footnotesize] [align=left] {$\displaystyle p$};
		% Text Node
		\draw (94,174) node [anchor=north west][inner sep=0.75pt]  [font=\scriptsize] [align=left] {$\displaystyle p$};

	\end{tikzpicture}
	
	\caption{How to fix the cases after a representative goes to the separator. Here, $v$ and $q$ are representatives and $v$ is going to be sent into the separator.}
	\label{fig:fixing}
\end{figure}

The main difference here is that Reed branches in the beginning and considers two cases. In the first case, none of the representatives goes to the separator, and in the second case at least one goes to the separator. In the second case, Reed's algorithm branches into at most $24k$ (upper bound for the number of subtrees). We want to avoid these branches and each time only branch into three cases. Assume we want to decide where to put $v$ (a representative with weight $w(v)$). Let $L$, $S$, and $R$ be the left, the separator, and the right sets, respectively. If we put $v$ into $L$ (or $R$), usually most of the vertices in its subtree will be in $L$ (or $R$) as well. In case $v$ goes to $S$, we release the other vertices of its tree to be searched again (as we described above and also as it is mentioned in the pseudocode). We have to mention that unlike Reed, we do not decide at the beginning if at least one vertex is going to the separator. Instead, we consider this case for every representative only when we handle that representative (remember that we handle the representative in a serial way).

\begin{theorem}\label{thm:main_one}
 Let $C_0 = 4 + 4\log 1.25 + \log 5 <7.60965$. For any $C>C_0$, there exists a 5-approximation algorithm that solves treewidth in time $\mathcal{O}(2^{Ck} n \log n)$, where $k$ is treewidth+1.
 \end{theorem}
\begin{proof}

Here, we include the pseudocodes of our algorithms (\ref{alg:nextrep-imp} and \ref{alg:dfs-imppp}), but before that we introduce the global variables.

\begin{tcolorbox}[width=\textwidth,colback={blue!4!white},title={\textcolor{black}{Global Variables}},colbacktitle=yellow!10!white,coltitle=black!50!white]    
   \begin{itemize}
       \item $G$: the input graph
       \item $root$: a fixed vertex to start all the DFS visits
       \item $W'$: the list containing all the representatives (that do not go into the separator)
       \item $\mathcal{W}$: the list containing the weight of every representative
       \item $decision$: the string that dictates where we should  place the representatives.
       \item $S$: the separator
       \item $k$: an upper bound on the capacity of the separator in the beginning.
       \item $k'$: an upper bound on the capacity of the separator at any given time.
   \end{itemize}
\end{tcolorbox}

\begin{algorithm}
	\DontPrintSemicolon
	\SetAlgoLined
	\textbf{Procedure INIT\_DFS()} \\
	\textbf{Declarations:}\\
	$_{^*}$ $\mathcal{R}$\_tree\_info $= \left(\text{list of vertices  }\mathcal{R}\_vertices, \text{ vertex } rep\_connected\right)$\;
	$_{^*}$ $s = \lfloor(1/4 - \epsilon/2)\frac{|V|}{k}\rfloor$ \tcp*{$s:$ the size bound for splitting off an $\mathcal{R}$-tree.}
	$_{^*}$ $decision = ``\,"$\tcp*{No representative should go into the separator}
	$_{^*}$\,\,\While{TRUE}{
	%$_{^*}$ temp\_decision = decision\;
	$_{^*}$ Initialize $W'$ and $\mathcal{W}$ with empty lists\tcp*{$W'$ is the (ordered) list of the representatives and $\mathcal{W}$ is the (ordered) list of their corresponding weights. $\mathcal{W} = \left(w[i]_{1\leq i \leq |W'|}\right)$, where $w[i]$ is the weight of the $i$-th representative.}
	$_{^*}$ $k' = k$\;
	$_{^*}$\,\,\For{\text{all } $v \in V$}{$_{^*}$ $color[v] = $ WHITE}
	$_{^*}$\,\,vertex\_\,rep = DFS\_VISIT(root)\;
	$_{^*}$ $x.\mathcal{R}\_vertices = x.\mathcal{R}\_vertices \circ vertex\_rep.\mathcal{R}\_vertices$, where $x$ is the last representative added to $W'$\tcp*{If the topmost subtree has smaller size than $s$, we hook it onto the last representative $x$.}
	$_{^*}$\,\,\If{$x.rep\_connected == nil$}{
	$_{^*}$ $x.rep\_connected = vertex\_rep.rep\_connected$
	}
	$_{^*}$\,\,\If{$k'==0$}{
	$_{^*}$\,\,\If{connected components of $G - S$ can form a balanced bipartition $(L, R)$}{
	$_{^*}$\,\,\Return{$(L, R)$}
		}
	}
	$_{^*}$\,\,\Else{try all partitions of $W'$ into $X$ and $Y$ to see if there exists a balanced separator.\;
	$_{^*}$\,\,\If{such a partition exists}{
	    $_{^*}$\,\,\Return{$(L, R)$}\tcp*{$L$ is $X$ along with all the vertices represented by vertices of $X$ (analogously $R$ corresponds to $Y$).}
	}
	}
	\tcp*{Now, we update the decision to look for the next possibility}
	$_{^*}$ $decision = $ lexicographically next string of at most the same length.\tcp*{For instance, $0011001$ follows $011000$ and $101$ follows $1001111$}
	$_{^*}$\,\,\If{no such string exists}{$_{^*}$\,\,\Return{``Treewidth is $> k$"}\;
	}
	$_{^*}$ Drop the ``0"s to the right of the rightmost ``1" in $decision$.\tcp*{We always remove the rightmost "0"s and add them back only if necessary (in the DFS-Visit() procedure)}
	}
	%$_{^*}$\,\,\Return{``No separator of size $\leq k$ exists"}\;	

	\textbf{End Procedure}
	\caption{The procedure to initialize the DFS}\label{alg:nextrep-imp}
\end{algorithm}

\begin{algorithm}
	\DontPrintSemicolon
	\SetAlgoLined
	\KwResult{$\mathcal{R}$tree\_info about the current partial $\mathcal{R}$tree}
	\textbf{Procedure DFS-Visit($u$)} \\
	$_{^*}$ $this = $ new $\mathcal{R}$\_tree\_info($\{u\}, nil$)\;
	$_{^*}$ color[$u$] = GRAY\;
	$_{^*}$\,\,\While{ \text{there exists $v$ adjacent to $u$ such that $color[v] ==$ WHITE}}{
			$_{^*}$ $child = \text{DFS\_VISIT}(v)$\;	
			$_{^*}$ $this.\mathcal{R}\_vertices = this.\mathcal{R}\_vertices\circ child.\mathcal{R}\_vertices$\; 
			$_{^*}$ color[v] = BLACK\;
		$_{^*}$ \If{$this.rep\_\,connected == nil$}{$_{^*}$ $this.rep\_\,connected =  child.rep\_\,connected$\tcp*{Each vertex remembers one of the possible representatives to which one of its descendants in the current $\mathcal{R}$tree is connected to and passes it to its parent.}
		}
	} 
	$_{^*}$ length $= |W'| + (k-k')$\;
	$_{^*}$\,\,\If{$|this.\mathcal{R}\_vertices|\geq s$}{
	$_{^*}$ \If{ len(decision) $< $length\tcp*{$0$ means place it left/right. Also, we assume that the decision is padded with unlimited zeros to the right of decision that we do not show after the last $1$, where $1$ means place into the separator}}{
	    $_{^*}$ $decision = decision \, \circ ``0"$\tcp*{where $\circ$ is string concatenation}
	}
	$_{^*}$ \If{ $decision[\text{length}] == 0$\tcp*{"$0$" means place it left/right}}{
	    $_{^*} W' = W'\cup\{this\}$\;
	    $_{^*}$ \Return{$(\emptyset, u)$}
	}
	$_{^*}$\,\,\Else{\tcp*{i.e. decision$[\text{length}]$ == 1 which means send $this$ into the separator}
	$_{^*}$ $k' = k'-1$\;
	$_{^*}$\,\,\If{$this.rep\_connected \neq nil$\tcp*{Handling the middle case of Figure~\ref{fig:fixing}}}{
	$_{^*}$ $w[this.rep\_connected] = w[this.rep\_connected] + |this.\mathcal{R}\_vertices|$ 
	$_{^*}$ \Return{$this$}\;
	}
	$_{^*}$\Else{\For{all $z \in this.\mathcal{R}nodes$}{$_{^*}$ $color[z] = WHITE$\tcp*{Either the left or the right case of Figure \ref{fig:fixing} has happened. In either case we do not need to worry.}}
	$_{^*}$\,\,\Return{$(\emptyset, nil)$}}
	}
	}
	\textbf{End Procedure}
	\caption{Main recursive procedure of DFS-Trees}\label{alg:dfs-imppp}
\end{algorithm}

The following section includes additional explanations on Algorithms \ref{alg:nextrep-imp} and \ref{alg:dfs-imppp} as well the proof of correctness.

\subsection{The Correctness of Our Algorithm}

We use a subroutine of Robertson-Seymour's 4-approximation algorithm \cite{robertson1995graph,reed_1997} to find a balanced separator of any given set $W$. We call this splitting ``split by $W$".  If we continue just by splitting $W$, we end up having quadratic dependence on $n$ in the running time in a worst case scenario. In order to avoid this, after every $\log_{\frac{3}{2}}k$ such splittings, we split based on $V$ (the reasoning can be found in Section~\ref{subse:correct_Reed}). Splitting by $V$ is handled by finding a collection of representatives and then splitting the representatives in a balanced way. As we mentioned earlier, at most $\frac{n}{C_{\epsilon}}$ vertices would be on the opposite side of their representatives. Hence, this way of splitting procedure splits the entire graph in a somewhat balanced way without actually deciding on every vertex (if there exists a balanced separator at all). Below we explain the latter part in more detail.

We do a DFS and the first time that we post-visit a vertex with at least $\frac{n}{C_{\epsilon}k}$ vertices in its subtree, we call that vertex a representative which represents the vertices of its subtree. Then, we decide what to do with that representative. 

Algorithm \ref{alg:nextrep-imp} initiates the process by calling the DFS on $G$ starting from a fixed vertex $root$. The reason that we fix the starting point is that we want to have a deterministic DFS search so that when we change our ``{\it decision}", the DFS follows the same order.

{\it Decision} is a string $\in \{0,1\}^*$ that tells us what to do after we find a representative. If one {\it decision} was not good (i.e., it did not give us a balanced separator), we check the next {\it decision}.

We start with the decision that no representative should go into the separator ($decision = ``\,"$). Notice that {\it decision} is a string. The reason is that we want to keep the leading $0$'s. Decision either is empty or ends with a $1$ (we stop at the right-most $1$ and assume the string is padded with an unlimited\footnote{In fact, $C_{\epsilon}k$ is sufficient.} number of $0$s to the right). If necessary, we keep adding those 0's to the right.  

If $decision[i] = 0$, we send the $i$-th representative left/right. If it was $1$, we send that vertex into the separator and do not consider it a representative anymore. If $|decision|$ is too short, this means it was supposed to be $0$ (because the string can be padded with unlimited 0s). Then, we add $0$ to the R.H.S. of the {\it decision} string and send this representative left/right. Sending left/right means just place it in $W'$, for now. At the end of this iteration, check all the possible ways of splitting $W'$ into $X$ and $Y$ to find a balanced separator. If such a split exist, return $(L, R)$, where $L = X \cup \{\text{any vertex represented by a representative in $X$}\}$ (analogously $R$ corresponds to $Y$). If none exist, check the next {\it decision} and start over.

Then, the while loop on line 6 checks the entire search space. For any fixed {\it decision}, we deterministically know what should happen to the representatives that we find. Each time we have a new decision, we set the variables back to their original values and restart searching for a balanced separator.

In line 12, we call the DFS\_Visit on $root$. DFS\_Visit returns a set of vertices (either the subtree that is too small to form a group represented by a representative or $\emptyset$ if they form a group) and the information whether any vertex in the set returned in the first argument is connected to another representative down below. (We will see later why we need this.) We refer to this pair of information as $\mathcal{R}$\_tree\_info (check line 3 of Algorithm~\ref{alg:nextrep-imp}).  
There is a chance that the last vertex post-visited (which is $root$) is not a representative, meaning that the size of its subtree is smaller than the threshold $s$. In this case, we hook the subtree rooted at {\it root} to the last representative that we have found (similar to part b of Figure \ref{fig:fixing}). This happens in line 13.

The main process happens in DFS\_VISIT. Whenever it post-visits a vertex, it checks whether its subtree is big enough. If the size of the current subtree hits the threshold, it forms a group and decides whether the current representative should go left/right or into the separator based on $decision$ (0 means send it into the separator and 1 means send it left/right). Otherwise, it passes the set of nodes to the parent along with a possible representative connected to its subtree down below.  

In line 28 of Algorithm \ref{alg:nextrep-imp}, we update the decision if the current decision has failed.

The next {\it decision} is the next string in lexicographical order with at most the same length. For example
\[
\begin{cases}
 1 & \text{follows } 0\\
 100111& \text{follows }100110111\\
 001011 & \text{follows } 0010101111111
\end{cases}
\] 

The reason that we do not check all the decision space is that if {\it decision} leads us to a failure, there is no reason to check a string with {\it decision} as one of its prefixes. Notice that every time we run out of the characters on the R.H.S. of {\it decision}, we just concatenate a `0' to the R.H.S. of it.

If $tw(G) \leq k$, then there exists a tree decomposition of $G$ namely $\mathcal{T}$ with width (at most) $k$. Then, based on Lemma~\ref{lem:cent}, there exists a strong centroid $x$ in $\mathcal{T}$. Both $\mathcal{T}$ and $x$ are unknown, but they do exist (if we knew, we already had a tree decomposition of width $k$ and a balanced separator). The connected components of $G\setminus B_x$ can form a bipartition ($L, R$) such that no part has more than $\frac{2}{3}n$ vertices. For the sake of argument, fix $\mathcal{T}$, $x$, $L$, and $R$. Every $W'\subseteq V$ is partitioned into parts in $L$ (name it $X$), parts in separator $S = B_x$, and parts in $R$ (name it $Y$). The weight of each part is at most $\left(\frac{2}{3} + \frac{1}{C_{\epsilon}}\right)n$. One of the branches picks $W'$. On the branch that picks $W'$, the separator $B_x$ lies between $X$ and $Y$. There is no guarantee that we find this separator, but the algorithm finds at least one separator of size $\leq k$ separating $X$ from $Y$ in a balanced way. Once more, up to $\frac{1}{C_{\epsilon}}n$ vertices are on the opposite side of their representative. Now, the larger side has at most $\frac{2}{3}n+ \frac{2}{C_{\epsilon}}n$ vertices and the smaller side has at least $\frac{1}{3}n - \frac{2}{C_{\epsilon}}n$ vertices. As we argued, we set $C_{\epsilon}\leftarrow\frac{4}{1-2\epsilon}$. This means that in a worst case scenario, our algorithm finds an $(\epsilon, 1-\epsilon)$-separation. Each part has a constant fraction of the vertices. Hence, we have a somewhat balanced separator.

The approximation ratio analysis is similar to what we described in Section~\ref{subse:correct_Reed}. We start with $|W| \leq 4k$. It can be written as $3k + k$. Let's call the second term, $k_{excess}$. Initially, $k_{excess} = k$. Each time we split $W$, the separator separates $W$ into two parts with the largest part having at most $2/3|W|$ vertices of $W$. We include the separator itself on both sides before recursing. This means 
\[|W| \leq \frac{2}{3}(3k + k_{excess}) + \underbrace{k}_\text{upper bound on the size of the separator} = 3k + \overbrace{\frac{2}{3}k_{excess}}^\text{new $k_{excess}$}.\]
Each time, the bound on $k_{excess}$ decreases by a factor of $2/3$. After $\log_{\frac{3}{2}} k$  iterations, $k_{excess}$ drops to zero. At this point, we do one split by $V$. We have to add the separator to both parts and this means $|W|$, which was upper bounded by $3k$, might go up to $4k$ once more. Now, we can do $\log_{\frac{3}{2}} k$ splits by $W$ again and continue as described.  

Notice that throughout the entire process $|W| \leq 4k$. Ultimately, in order to merge the two tree decompositions that we found for two subproblems, we add the separator to both sides. This means the largest bag can have up to $5k$ vertices and hence, the approximation ratio is $5$.

\section{Running Time of Our Algorithm}
Let $G$, $k$, and $t$ be the initial graph, the bound on the size of the separator, and the number of representatives, respectively. $t$ is at most $C_{\epsilon}k$ since the cut-off threshold for the volume of the subtrees was $\frac{n}{C_{\epsilon}k}$. While proceeding with the algorithm at each step, let $G'$, $k'$, and $t'$ be the current graph, the current bound on the size (capacity) of the separator, and a bound on the number of representatives still to find, respectively. Each time we send some vertex to the separator, we decrease the capacity by one. We assume a worst case, where the number of $\mathcal{R}$\_trees is $t$, and the separator found has size $k$ (other cases only speed up the running time). Then, the recurrence for the running time to find a separator of size at most $k$ in $G$ is:
\begin{equation}\label{eq:rec}
T(t, k) \leq T(\left\lceil t\cdot\frac{k-1}{k}\right\rceil, k-1) + 2T(t-1, k) + Qkn + \mathcal{O}(1), \text{ for } t,k >0, 
\end{equation}  
where $Q$ is the constant factor of the DFS. On the R.H.S., the first term handles the case where the representative goes to the separator. Therefore, $k$ decreases by 1, and the number of subtrees becomes at most $t\cdot\frac{k-1}{k}$ (delete that vertex and continue the DFS). The second term handles the case that the current representative does not go into the separator but left or right. Here, the capacity of the separator is unchanged but the number of subtrees decreases by 1. The third term is the upper bound of the exact running time of the DFS. And the last term $\mathcal{O}(1)$ is the overhead to make the calls. The base cases of the recurrence:
\begin{equation}\label{eq:base}
\nonumber
\begin{tabular}{ll}
$T(0, k)  \leq Q(k+1)kn$ & $\text{to do } k \! + \! 1 \text{ DFSs. We find } k \text{ paths from $X$ to $Y$,}$ \\
	& and search for augmenting paths. \\
$T(t, 0)  \leq Qkn$ & $\text{We have to test whether } S \text{ is indeed a separator}.$ 
\end{tabular}
\end{equation}
The recurrence in \ref{eq:rec} might seem hopeless, so we simplify it (generous bound).
\begin{equation}\label{eq:time}
\tcboxmath{T(t, k) \leq T(t, k-1) + 2T(t-1, k) + Qkn + \mathcal{O}(1),\,\, t,k>0}
\end{equation}
Now, we have to solve this recurrence. Our recursion tree starts from the root $T(t,k)$ and has two children $T(t, k-1)$ and $T(t-1, k)$, left and right, respectively. This is an unbalanced binary tree. Each strand terminates when one of the arguments of $T(\cdot\, ,\cdot)$ becomes zero. Each time we choose the left branch (putting one representative into the separator), we decrease $k$ by 1. Otherwise (putting the representative and its subtree to the right or left set of the separator), we decrease $t$ by 1 and multiply the value by $2$. Let $\#(t-i,0)$ be the (worst case) number of leaves with the first argument $t -i$ and the second argument $k = 0$, for $0\leq i < t$ (analogous notation for $\#(0,k-j)$ for $0 \leq j < k$). Observe that $\#(t-i,0) = \binom{k+i}{i}$, and $\#(0,k-j) = \binom{t+j}{j}$. The first two terms of Equation \ref{eq:time} can be computed at the leaves and the other two terms are spent in every vertex of the recursion tree. The reason is that for the first two terms, we need the results of the children and recursively everything relies on the results of the leaves.  Each time we decide to consider a subproblem (any internal node corresponds to a subproblem), we have to find the actual separator. This step takes $Qkn$ time. Also, in order to make the recursive calls, we spend $\mathcal{O}(1)$ time.

Now, we compute the first part (the first two terms).
\begin{align*}
&\sum_{i = 0}^{t-1}\left(\#(t-i,0)\hspace{-1.5em}\underbrace{2^i}_\text{i right branches}\hspace{-1.5em}T(t-i, 0)\right)
+ \sum_{j = 0}^{k-1}\left(\#(0,k-j)\hspace{-1.5em}\underbrace{2^{t}}_\text{$t$ right branches}\hspace{-1.5em}T(0, k-j)\right)
\\&\leq  \sum_{i = 0}^{t}\left(\binom{k+i}{i}2^i\underbrace{Qk^2n}_\text{Eq. \ref{eq:base}}\right) 
+ \sum_{j = 0}^{k-1}\left(\binom{t+j}{j}2^{t}\underbrace{Qkn}_\text{Eq. \ref{eq:base}}\right)
\\&\leq  Qk^2n\sum_{i = 0}^{t}\binom{k+i}{k}2^i+Qkn2^{t}\sum_{j = 0}^{k-1}\binom{t+j}{j}
\\&< Qk^2n\sum_{i = 0}^{t}\binom{k+t}{k}2^i+Qkn2^{t}\binom{k+t}{k-1}
\\&\leq Qk^2n\binom{k+t}{k}\left(2^{t+1}-1\right)+Qk^2n2^{t}\binom{k+t}{k} 
\end{align*}
Now, we have to compute the second part of Equation \ref{eq:time} where we should look at every internal vertex of the tree. We have $\binom{t+k}{k}-1$ internal vertices (in worst-case), and in each vertex with value $k'$, we spend at most $Qk'kn + \mathcal{O}(1) \leq Qk^2n(1+\mathcal{O}(\frac{1}{n}))$. Hence, 
\begin{equation}
\begin{split}\label{eq:ruTi}
T(t, k) &\leq \left(2^{t+1}-1\right)Qk^2n\binom{k+t}{k} + Qk^2n(1+\mathcal{O}(\frac{1}{n}))\binom{t+k}{k}
\\& = Qk^2 n \left(2^{t+1}+\mathcal{O}\left(\frac{1}{n}\right)\right)\binom{k+t}{k}
\end{split}
\end{equation}
Notice that $T(\cdot\, , k)$ is monotonic (due to the definition of $T$) and use the fact that $t \leq C_{\epsilon}k$.
Now, we simplify Equation \ref{eq:ruTi} by bounding $T(t, k)$ with $T(C_{\epsilon}k, k)$ and using Striling's approximation.
\begin{align*}
    T\left(t, k\right)&\leq T\left(C_{\epsilon}k, k\right)\leq Qk^2n\left(2^{C_{\epsilon}k+1} + \mathcal{O}\left(\frac{1}{n}\right)\right)\binom{\left(C_{\epsilon}+1\right)k}{k}
    \\&\leq Qk^{\frac{3}{2}}n\left(2^{C_{\epsilon}k+1} + \mathcal{O}\left(\frac{1}{n}\right)\right) \sqrt{\frac{C_{\epsilon}+1}{2\pi C_{\epsilon}}}\left(1+\frac{1}{C_{\epsilon}}\right)^{C_{\epsilon}k}\left(C_{\epsilon}+1\right)^k
\end{align*}

Now, we compute the running time ($T_V$), when we split based on $V$. We assume a worst case split of $\epsilon$ to $1-\epsilon$ and worst case separator of size $k$.
\begin{align*}\label{eq:final_time}
&T_{V}(n, k) = T_{V}(\epsilon n+k, k) + T_{V}((1-\epsilon)n+k, k) + T(t,k)
\\&\leq T_{V}(\epsilon n+k, k) + T_{V}((1-\epsilon)n+k, k)  
\\&\quad + Qk^{\frac{3}{2}}n\left(2^{C_{\epsilon}k+1} + \mathcal{O}\left(\frac{1}{n}\right)\right) \sqrt{\frac{C_{\epsilon}+1}{2\pi C_{\epsilon}}}\left(1+\frac{1}{C_{\epsilon}}\right)^{C_{\epsilon}k}\left(C_{\epsilon}+1\right)^k
\\& \leq Qk^{\frac{3}{2}}\left(2^{C_{\epsilon}k+1} + \mathcal{O}\left(\frac{1}{n}\right)\right) \sqrt{\frac{C_{\epsilon}+1}{2\pi C_{\epsilon}}}\left(1+\frac{1}{C_{\epsilon}}\right)^{C_{\epsilon}k}\left(C_{\epsilon}+1\right)^k\frac{1}{\epsilon}n\ln n
\end{align*}

In the above equation, we use the corollary to the following lemma (Lemma~ \ref{lem:runtim}). Note that the reason we add $k$ to both recursive calls is that we add the separator to both subproblems.
\begin{lemma}\label{lem:runtim}
Assume $0 < \epsilon \leq \frac{1}{2}$, $0<c' \leq c$, $2 \leq k$,  and  $n_1 + n_2 =n$. Then the recurrence
\[ f(n+k)  \leq 
\begin{cases}
 c' (n+k) & \mbox{if $n-k \leq 4k$} \\
 f(n_1 + k) + f(n_2 + k) + c (n+k)  & \mbox{otherwise,} 
\end{cases}
\]
where $\frac{1}{2} n \leq n_1 \leq (1 - \epsilon)n$ has a solution with $ f(n+k) \leq \frac{c}{\epsilon}  n \ln n  - ck $, for $n \geq 2k$.
\end{lemma}

\begin{proof}
	Case 1: $2k \leq n \leq 4k$. Then $n \geq 4$ implying $\ln n > 1$ and
	\[f(n+k) \leq c(n+k)  <  2cn  - ck < \frac{c}{\epsilon}  n \ln n  - ck.\]
	
	\noindent
	Case 2: $n \geq 4k$ and $n_2  \geq 2k$.
	\begin{eqnarray*}
		f(n+k) & \leq & f(n_1 + k) + f(n_2 + k) + c (n+k) \\
		& \leq & \frac{c}{\epsilon} \left( n_1 \ln n_1 + (n - n_1) \ln (n - n_1) \right) - 2ck + c (n+k) \\
		& < & \frac{c}{\epsilon} \left((1 - \epsilon) n (\underbrace{\ln (1-\epsilon)}_{< - \epsilon} + \ln n) + \epsilon n (\ln \epsilon + \ln n)\right)  + c (n-k) \\
		&  < & \frac{c}{\epsilon} n \ln n - c (1-\epsilon) n + c n \ln \epsilon + c (n-k) \\
	%	& \leq & \frac{c}{\epsilon} n \ln n - cn +  c \epsilon n + c n \ln \epsilon + c (n-k) \\
		& \leq & \frac{c}{\epsilon} n \ln n + (\underbrace{\epsilon \; \; \; + \; \; \: \ln \epsilon}_{< 0 \mbox{ for $\epsilon \leq \frac{1}{2}$}}) \, cn - ck \\
		& \leq & \frac{c}{\epsilon} n \ln n - ck 
	\end{eqnarray*}
	
	\noindent
	Case 3:  $n \geq 4k$ and $n_2 = n - n_1 < 2k$.
	\begin{eqnarray*}
		f(n+k) & \leq &  f(n_1 + k) + f(n_2 + k) + c (n+k) \\
		& \leq & \frac{c}{\epsilon} n_1 \ln n_1 - ck + c'(n_2 + k)  + c (n+k) \\
		& \leq & \frac{c}{\epsilon} (1 - \epsilon) n (\ln (1-\epsilon) + \ln n) - ck + c'(n_2 + k)  + c (n+k) \\
		& < & \frac{c}{\epsilon} (1 - \epsilon) n (-\epsilon + \ln n) - ck + c'(n_2 + k)  + c (n+k) \\
		& \leq & \frac{c}{\epsilon} n \ln n - ck - cn \ln n - c (1 - \epsilon) n + c'(n_2 + k)  + c (n+k) \\
		& < & \frac{c}{\epsilon} n \ln n - ck 
	\end{eqnarray*}
	The last inequality is true, because $n \geq 4k \geq 8$ implying $\ln n >2$, and $n+ k < 3k\leq \frac{3}{4}n$.
\end{proof}

\begin{corollary}
Under the conditions of the Lemma \ref{lem:runtim},
\[  f(n) \leq \frac{c}{\epsilon} n \ln n .\]
\end{corollary} 
Now, the total running time of the algorithm ($T_t$) is:
\begin{align}\label{eq:totalruntimeeditted}
T_t(n,k) \leq Qk^{\frac{3}{2}}\left(2^{C_{\epsilon}k+1} + \mathcal{O}\left(\frac{1}{n}\right)\right) \sqrt{\frac{C_{\epsilon}+1}{2\pi C_{\epsilon}}}\left(1+\frac{1}{C_{\epsilon}}\right)^{C_{\epsilon}k}\left(C_{\epsilon}+1\right)^k\frac{\left(1+\log_{\frac{3}{2}} k\right)}{\epsilon}n\ln n,
\end{align}
where it takes $\log_{\frac{3}{2}}k$ steps so that the $k_{excess}$ drops to zero (that is when we need to split by $V$ once more).\\
As we mentioned in Section \ref{sec:divide}, $C_{\epsilon} = \frac{4}{1 - 2 \epsilon}$. We plug that into Equation \ref{eq:totalruntimeeditted}.
\begin{equation}\label{eq:lat_time}
T_t(n, k) \leq Qk^{\frac{3}{2}}\left(2^{\frac{4k}{1-2\epsilon}+1} + \mathcal{O}\left(\frac{1}{n}\right)\right) \sqrt{\frac{5 - 2\epsilon}{8\pi}}\left(1.25 -  \frac{\epsilon}{2}\right)^{\frac{4k}{1-2\epsilon}}\left(\frac{4}{1-2\epsilon}+1\right)^k\frac{\left(1+\log_{\frac{3}{2}} k\right)}{\epsilon}n\ln n,
\end{equation}

    In Equation~\ref{eq:lat_time}, let $\epsilon \rightarrow 0$.
    \begin{align}
    \lim_{\epsilon\rightarrow 0} T_t(n, k) &\leq\frac{1}{\epsilon} \sqrt{\frac{5}{2\pi}}Qk^{\frac{3}{2}}\left(\log_{\frac{3}{2}} k\right) 2^{\left(4 + 4\log 1.25 + \log 5\right)k} n \ln n
    \end{align}
\end{proof}

For instance, by setting $\epsilon = 10^{-2}$,we have 
\[T_t(n, k) \leq 90Q k^{1.5}\log_{3/2} k \, 2^{7.718k} n \ln n,\]
and by setting $\epsilon = 10^{-3}$, we get
\[
T_t(n, k) \leq 893Q k^{1.5}\log_{3/2} k \, 2^{7.621k} n \ln n
\]
\begin{corollary}
There exists a 5-approximation algorithm for treewidth that runs in time $\mathcal{O}\left(2^{7.61k}n\log n\right).$
\end{corollary}
 \section{Conclusion}
In this paper, we have given a detailed analysis of Reed's  treewidth approximation algorithm\cite{reed1992finding}. We have shown that it runs in time $\mathcal{O}(2^{24k}k!n\log n)$, where $k$ is treewidth+1. Furthermore, we have shown that it is a 5 or 7 approximation, depending on how frequently we split based on $V$.

Then, we have given our improved algorithm which runs in time $\mathcal{O}(2^{7.61k} n \log n)$, with the same approximation ratio.
 
 Our main goal was to obtain a small coefficient of $k$ in the exponent to make the algorithm more applicable. We think it is still possible to further improve it. Also, trying to come up with a better approximation ratio is worthwhile while maintaining the same dependence on $k$ (a 2-approximation \cite{korhonen2021single} is known but the coefficient of $k$ in the exponent is larger). 
 
 Finally, another direction for future work could be giving hardness results on the approximation ratio.

\bibliographystyle{abbrvurl}

%\bibliography{references}  %%% Uncomment this line and comment out the ``thebibliography'' section below to use the external .bib file (using bibtex) .





% \end{thebibliography}

\end{document}